\documentclass{article}
\usepackage[utf8]{inputenc}
\usepackage{graphicx,amsmath,amssymb,bm,amsthm,comment,color,appendix,bbm}
\usepackage{a4wide} 
\usepackage{ulem}
\newcommand{\N}{\mathbb{N}}
\newcommand{\Z}{\mathbb{Z}}
\newcommand{\R}{\mathbb{R}}
\newcommand{\C}{\mathbb{C}}

\newcommand{\cA}{\mathcal{A}}
\newcommand{\mc}[1]{\mathcal{#1}}

\newcommand{\e}{\varepsilon}

\newcommand{\pphi}{\varphi}

\newcommand{\re}{{\rm Re}\hskip 1pt }
\newcommand{\im}{{\rm Im}\hskip 1pt }
\newcommand{\ord}{{\mathcal O}}

\newcommand{\su}{\mathrm{\, SU}}
\newcommand{\cdt}{{\circledast}}
\newcommand{\ope}[1]{\operatorname{#1}}

\newcommand{\be}{\begin{equation}}
\newcommand{\ee}{\end{equation}}
\newcommand{\ben}{\begin{equation*}}
\newcommand{\een}{\end{equation*}}

\newcommand{\LA}{\Lambda_{\textrm A}}
\newcommand{\LN}{\Lambda_{\textrm N}}
\newcommand{\muN}{\mu_{\textrm N}}
\newcommand{\muA}{\mu_{\textrm A}}
\newcommand{\mN}{m_{\textrm N}}
\newcommand{\mA}{m_{\textrm A}}
\newcommand{\gammaN}{\gamma_{\textrm N}}

\newcommand{\NN}{\textrm{NN}}
\newcommand{\AAa}{\textrm{AA}}
\newcommand{\NA}{\textrm{NA}}

\newcommand{\dip}{\displaystyle}

\newcommand*\oline[1]{%
  \vbox{%
    \hrule height 0.5pt
    \kern0.25ex
    \hbox{%
      \kern-0.1em
      \ifmmode#1\else\ensuremath{#1}\fi
      \kern-0.1em
    }
  }
}
\newtheorem{theorem}{Theorem}
\newtheorem{lemma}{Lemma}[section]
\newtheorem{proposition}[lemma]{Proposition}
\newtheorem{definition}[lemma]{Definition}
\newtheorem{remark}[lemma]{Remark}

\newtheorem{example}[lemma]{Example}
\numberwithin{equation}{section}
\newcounter{Condition}

\newtheorem{condition}[Condition]{Condition}

\title{Two-level Adiabatic Transition Probability for Small Avoided Crossings generated by tangential intersections}
\author{Kenta Higuchi and Takuya Watanabe} 
\date{\today}

\begin{document}

\maketitle

\begin{abstract}
In this paper, the asymptotic behaviors of the transition probability for two-level avoided crossings are studied under the limit where two parameters (adiabatic parameter and energy gap parameter) tend to zero. 
This is a continuation of our previous works where avoided crossings are generated by tangential intersections and obey a non-adiabatic regime. 
The main results elucidate not only the asymptotic expansion of transition probability but also a quantum interference caused by several avoided crossings and a coexistence of two-parameter regimes arising from different vanishing orders.  
\end{abstract}



\section{Introduction}
In quantum mechanics, especially in the quantum chemistry, the adiabatic approximation 
and the Born-Oppenheimer approximation are widely used. The adiabatic theorem, the motivation of these approximations, asserts that 
in the slowly varying Hamiltonian the quantum effect like the transition between the energy-levels hardly occurs.
From this point of view, it is important to accurately describe how much slowing down the variation shrinks the transition probability.

In this paper, we study a mathematical model such that the transition probability is not always small even in case of the adiabatic approximation. Since the transition probability intuitively depends on the size of the smallest gap between energy-levels, the approaching (resp. receding) speed 
to (resp. from) the smallest gap, and the quantum interference, 
we consider asymptotic behavior in a two-parameter singular limit $h,\e\to+0$ of solutions to the time-dependent Schr\"odinger equation
\be\label{eq:OurEq}
ih\frac d{dt}\psi(t)=H(t;\e)\psi(t),\quad t\in\R.
\ee
Here, the Hamiltonian $H(t;\e)$ is given as a $2\times2$ matrix-valued function
\be\label{H_matrix}
H(t;\e):=
\begin{pmatrix}
V(t)&\e\\\e&-V(t)
\end{pmatrix},
\ee
where $V(t)$ is a real-valued smooth function and $h,\e$ are small positive parameters. Its two eigenvalues are
\ben
E_\pm (t;\e) = \pm \sqrt{V(t)^2 + \e^2}.
\een
In this model, the ratio $t/h$ is interpreted as the time variable, $E_\pm(t;\e)$ are the two energy-levels of $H(t;\e)$, and the adiabatic limit $h\to0$ corresponds to the slow variation of the Hamiltonian $H(t;\e)$ compared with the time. 
Note that each solution $\psi=\psi(t)$ takes values in $\C^2$, and its norm remains constant for all $t\in\R$:
\ben
\left\|\psi(t)\right\|_{\C^2}^2=\left|\psi_1(t)\right|^2+\left|\psi_2(t)\right|^2=\text{(Const.)},\qquad
\psi(t)=
\begin{pmatrix}
    \psi_1(t)\\\psi_2(t)
\end{pmatrix}.
\een

According to the adiabatic theorem, one expects that for a solution $\psi(t)$ to \eqref{eq:OurEq}, the projection $\Pi_-(t;\e)\psi(t)$ onto the eigenspace associated with $E_-(t;\e)$ is ``small"  for every $t\in\R$ if $\psi(t_0)$ belongs to the eigenspace associated with $E_+(t_0;\e)$ at some $t_0\in\R$. 
More simply, we can say that the adiabatic theorem asserts the smallness of the transition probability.
Here, we call 
\be\label{eq:defTP}
P(\e,h):=\lim_{t\to+\infty}\left\|\Pi_-(t;\e)J_\ell^+(t)\right\|_{\C^2}^2
\ee
the transition probability, where $J_\ell^+$ is the normalized solution 
such that 
\begin{equation*}
    \lim_{t\to-\infty}\left\|\Pi_-(t;\e)J_\ell^+(t)\right\|_{\C^2}=0
\end{equation*}
with $\left\|J_\ell^+(t)\right\|_{\C^2}=1$ 
(this solution will be introduced in Appendix~\ref{app_jost}). These limits exist under suitable conditions on $V$ near infinity (Condition~\ref{condi_1} in this paper).
Note that $\left\|J_\ell^+(t)\right\|_{\C^2}^2=\sum_\pm \left\|\Pi_\pm(t;\e)J_\ell^+(t)\right\|^2_{\C^2}=1$ holds for any $t\in\R$. This is the reason why one can regard $\left\|\Pi_\pm(t;\e)J_\ell^+(t)\right\|^2_{\C^2}$ as transition/reflection probabilities.

As long as $\e>0$, the two energy-levels $E_\pm(t;\e)$ are smooth functions of $t$, and never intersect with each other: 
\be\label{eq:E-gap}
\inf_{t\in\R}\left|E_+(t;\e)-E_-(t;\e)\right|=\inf_{t\in\R}2\sqrt{V(t)^2+\e^2}\ge 2\e>0.
\ee
This quantity called the energy-gap is bounded from below by $2\e$ even if $V$ vanishes at some point. This phenomenon occurring near each zero of $V$ is called an avoided crossing. 
The simplest case $V(t)=vt$ with a positive constant $v$ is investigated individually by L.D. Landau and C. Zener in 30's \cite{La65,Ze32_01}.
The transition probability 
\be\label{LZ_formula_new}
P(\e,h)=\exp\left(-\frac{\pi\e^2}{vh}\right)
\ee
for this case is known as the Landau-Zener formula. This is exact and true for any positive $\e,h$. 
For fixed $\e>0$, this formula implies that the transition probability is exponentially small with respect to $h>0$. 
There are many results generalizing the Landau-Zener formula. 
Under some analyticity condition, such an exponential decay estimate is obtained even in case of more general Hamiltonian, for example operator-valued unbounded Hamiltonians \cite{Co03, Co04, Jo94_01, JoKuPf91_01, Ma94_01}, while a smoothness condition without an analyticity yields a polynomial decay estimate \cite{Ka50_01}. 
Note that in the general setting, the condition of the energy-gap is replaced with the gap condition, which mandates that the spectrum is decomposed into a disjoint union of two subsets and that the distance between them is positive. The history of these generalizations can be consulted in the survey \cite{HaJo05_01} and in the books \cite{Ha94_01, Te03_01}.

The transition probability may become larger when the energy-gap is also small. In our model, this situation occurs if $V(t)$ vanishes at some $t$ and if $\e$ (see \eqref{eq:E-gap}) is sufficiently small compared with $h$. 
One observes from Landau-Zener formula \eqref{LZ_formula_new} that the transition probability is 
small and the adiabatic approximation is reasonable if $\e\gg h^{1/2}$. However, one also observes that it is almost one if $\e\ll h^{1/2}$. The former situation is called the adiabatic regime, and the latter the non-adiabatic regime \cite{CoLoPo99_01, Ro04_01, WaZe21_01}.

The leading term of the transition probability is given by the same formula as \eqref{LZ_formula_new} by replacing $v$ with $|V'(0)|$ when $V(t)$ vanishes only at $t=0$ and $V'(0)\neq0$, namely, the situation that $V(t)$ and $-V(t)$ intersect transversely at $t=0$ (see the work \cite{Jo94_01} and also its microlocal version \cite{CoLoPo99_01}). From the viewpoint of the energy-levels, the approaching/receding $|E_+(t;\e)-E_-(t;\e)|-2\e=2(\sqrt{V(t)^2+\e^2}-\e)$ near a transversal crossing of $\pm V$ is of order $|t|$. 

In the tangential case $V'(0)=0$, the transition probability is studied by one of the authors under the condition $\e\gg h^{m/(m+1)}$ corresponding to the adiabatic regime, where $m$ stands for the vanishing order of $V$ at $t=0$ as in \cite{Wa06_01,Wa12_01} (equivalently, $|E_+(t;\e)-E_-(t;\e)|-2\e$ is of order $|t|^m$). In this case, transition probability is exponentially small as $h\e^{-(m+1)/m}$ tends to 0. The analyticity of $V$ and the adiabatic regime condition are necessary for applying the exact WKB method. In fact, the ``complex crossing points" of the energy-levels, which are the zeros of $E_+(t;\e)-E_-(t;\e)$ on the complex plane and are called turning points in the WKB method, are essential for this case. The adiabatic regime condition implies that these complex crossing points are not too close to each other.  

On the other hand, the situation corresponding to the non-adiabatic regime $\e\ll h^{m/(m+1)}$ is studied by the other author \cite{Hi23_01}. He applied other classical method (which is also recently used for other problem \cite{AFH22_01}) to a little bit more general setting. The transition probability is almost one as in the Landau-Zener formula only when $m$ is odd, and that it is still small of order $\e h^{-m/(m+1)}$ when $m$ is even. 

One of other generalizations is the existence of several avoided crossings. Following the classical probability theory, one may think that the transition probability is obtained by multiplying and summing the non-negative ``local transition probability" around each avoided crossing. However, as well as other quantum situations, only a complex-valued probability amplitude is associated with each avoided crossing. Then the ``total" probability amplitude is given by multiplying and summing them, and the transition probability is its absolute square. 
This phenomenon has been also treated \cite{JoMiPf91_01,Wa12_01,WaZe21_01}. 

This paper is a continuation of the authors' previous works in the viewpoint of dealing with several avoided crossings generated by tangential intersections with different vanishing orders in the non-adiabatic regime. 
Our first result, Theorem~\ref{mainthm}, concerns several tangential avoided crossings in the non-adiabatic regime, that is, $\e\ll h^{m/(m+1)}$ with $m$ the maximum among the avoided crossings. It shows that the transition probability is almost one when the number of odd avoided crossings is odd and that it is small of order $\e h^{-m/(m+1)}$ when the number is even. The effect of the quantum interference appears in the coefficient of the term of order $\e h^{-m/(m+1)}$. In Formula~\eqref{prefactor}, the second term describes the quantum interference while the first term is given by the sum of absolute square of the local transition probability amplitudes. In particular, this coefficient vanishes in some cases. We also show some concrete models (see Remark~\ref{delta_vanish} and Examples~\ref{2avo} and \ref{3avo}).

One notices that the border of the parameter regimes for each avoided crossing depends on the vanishing order $m$. Consequently, there are parameter regimes which is adiabatic for some avoided crossings and non-adiabatic for the others when there are several tangential intersections of $V(t)$ and $-V(t)$. 
Our second result , Theorem~\ref{2ndthm}, concerns this situation, and shows that the leading term of the transition probability depends on the parity of the number of odd avoided crossings in the non-adiabatic regime. Since the local probability amplitude around an avoided crossing in the non-adiabatic and adiabatic regime has already been computed in Theorem~\ref{mainthm} and in the previous work \cite{Wa12_01}, Theorem~\ref{2ndthm} is obtained by combining them. 
The novelties of this paper are to examine precisely the transition probability in the intermediate regime, where the non-adiabatic regime and the adiabatic one coexist, and to elucidate a possibility of ``switching of the transition probability" by varying two parameters $\e, h$ continuously without changing $V(t)$ as in Example \ref{ex:intermediate}.  
Note that the situation neither adiabatic nor non-adiabatic regime, namely, $\e\sim h^{m/(m+1)}$ for some $m\ge2$, has not been treated yet, although the case for $m=1$ has done \cite{Ha91_01}.

Our proof is based on the classical method. We first introduce the Jost solutions $J_\ell^\pm=J_\ell^\pm(t;\e,h)$ and $J_r^\pm=J_r^\pm(t;\e,h)$ admitting the asymptotic behavior \eqref{eq:Jost-sol} at infinity, and in particular, $J_\ell^+$ satisfies \eqref{eq:defTP} (see Appendix~\ref{app_jost} for the construction). Then the total transition probability amplitude and the transition probability are $s_{21}(\e,h)$ and the square of its modulus, where $s_{21}(\e,h)$ stands for the $(2,1)$-entry of the scattering matrix $S(\e,h)$ defined by
\ben
(J_\ell^+(t;\e,h),J_\ell^-(t;\e,h))=(J_r^+(t;\e,h),J_r^-(t;\e,h))S(\e,h).
\een
Note that one has
\ben
\Pi_-(t;\e)J_\ell^+(t;\e,h)- s_{21}(\e,h)J_r^-(t;\e,h)\to0\quad\text{as }t\to+\infty.
\een
To study the entries of $S(\e,h)$, we continue the solutions $J_\ell^\pm$ from $-\infty$ to $+\infty$. More precisely, we construct solutions which approximately belong to the eigenspace associated with $E_\pm(t;\e)$ away from any avoided crossings, and compute the transfer matrices between the bases consisting of such solutions. 
The transfer matrix is almost diagonal when there is no avoided crossing between two points. 
Thus, the transfer matrix $T_k$ across each avoided crossing near $t_k$ is crucial to obtain the transition probability.
The four entries of $T_k$ are the probability amplitudes of the local transition at the vanishing point $t_k$. 

The asymptotic behavior of $T_k$ around each avoided crossing near $t_k$ is given in Theorem~\ref{thm:Connection}. As we mentioned above, the exact WKB solutions used in the previous work \cite{Wa12_01} concerning avoided crossings generated by tangential intersection are no longer valid in the non-adiabatic regime. The solutions are constructed in Section~\ref{sec_const_sol} by the method of successive approximations (MSA for short) due to the previous works \cite{AFH22_01,Hi23_01}. 
For example, the $(1,2)$ and $(2,1)$-entries of $T(\e,h)$ correspond to the local transition probability amplitude from $E_+$ to $E_-$ and from $E_-$ to $E_+$ when the vanishing order $m$ is odd and $V(t)(t-t_k)\ge0$ near $t_k$. The leading term of them is given by applying the degenerate stationary phase method (Lemma~\ref{lem:Int-esti}) to the oscillatory integral \eqref{T_asym_comp_04}, where the derivative of the phase function $\mp 2\int_0^t V(r)dr$ off-course has a zero of the same order as $V$.

This paper is organized as follows. In Section~\ref{sec_results}, we make precise the definitions and settings, and state our main results Theorems~\ref{mainthm} and \ref{2ndthm}. 
We construct the solutions by the method of successive approximations (MSA) in Section~\ref{sec_const_sol}, and prove the connection formulas Theorem~\ref{thm:Connection} and Proposition~\ref{prop_Tkk1} by using these solutions in Section~\ref{sec_connection}.
Finally, we will complete the proofs in Section~\ref{sec_end_proof}. 
To obtain the product of $2n+1$ matrices of $\su (2)$, we employ an algebraic formula shown in Appendix~\ref{app_alg_lem}.

\section{Results}\label{sec_results}

\subsection{Assumptions and main result}

As mentioned in the introduction, we focus on the non-adiabatic regime and work under the $C^\infty$-category without any assumption on the analyticity. 
We notice that the assumption on $V(t)$ and the setting of the problem are sightly different from the previous work \cite{WaZe21_01} 
but the definitions of the transition probability in the series of our works are the same. We first assume the following:

\begin{condition}\label{condi_1}
 The function $V(t) \in C^\infty(\R;\R)$ has a limit $V_r \in \R\setminus \{0\}$ (resp. $V_\ell\in \R\setminus \{0\}$)  as $t \to +\infty$ (resp. $-\infty$), and satisfies
\ben
V-V_r\in L^1([0,+\infty)),\quad V-V_\ell\in L^1((-\infty,0]),\quad V'\in L^1(\R).
\een
\end{condition}
For simplicity, we assume $V_r>0$.
Based on the argument in Appendix \ref{app_jost} under Condition \ref{condi_1}, one sees the unique existence of Jost solutions $J_\bullet^{\pm}(t)$ ($\bullet \in \{\ell,r\}$) which satisfy the asymptotic conditions:
\be\label{eq:Jost-sol}
\begin{aligned}
 &J_r^+ (t) \sim \exp \left[-\frac{it}{h}\sqrt{V_r^2+\varepsilon^2} \right]\left(
\begin{array}{c}
\cos{\theta_r}\\
\sin{\theta_r}
\end{array}
\right) 
 &&\textrm{as} \,\,\,
 t\to +\infty,\\
 &J_r^- (t) \sim \exp \left[+\frac{it}{h}\sqrt{V_r^2+\varepsilon^2} \right]\left(
\begin{array}{c}
-\sin{\theta_r}\\
\cos{\theta_r}
\end{array}
\right)
 &&\textrm{as} \,\,\,
 t\to +\infty,\\
 &J_\ell^+ (t) \sim \exp \left[-\frac{it}{h}\sqrt{V_\ell^2+\varepsilon^2}
 \right]\left(
\begin{array}{c}
\cos{\theta_\ell}\\
\sin{\theta_\ell}
\end{array}
\right) 
 &&\textrm{as} \,\,\,
 t\to -\infty,\\
 &J_\ell^- (t) \sim \exp \left[+\frac{it}{h}\sqrt{V_\ell^2+\varepsilon^2} \right]\left(
\begin{array}{c}
-\sin{\theta_\ell}\\
\cos{\theta_\ell}
\end{array}
\right)
 &&\textrm{as} \,\,\,
 t\to -\infty,
 \end{aligned}
\ee
where $\tan{2\theta_{\bullet}}=\varepsilon/V_{\bullet}$ with $0<\theta_{\bullet} <\pi/2$ (equivalently determined by $\theta_{\bullet}=\arctan(\e^{-1}(\sqrt{V_{\bullet}^2+\e^2}-V_{\bullet}))$). Note that $\theta_\bullet$ never coincides with $\pi/4$ for small $\e$ since one has $\theta_\bullet=\ord(\e)$ when $V_\bullet>0$, and $\pi/2-\theta_\bullet=\ord(\e)$ when $V_\bullet<0$. 
The pairs $(J_r^+, J_r^-)$ and $(J_\ell^+, J_\ell^-)$ form bases of the solution space. 
Each of them corresponds to one of the eigenvalues $\pm\sqrt{V_r^2+\e^2}$ and $\pm\sqrt{V_\ell^2+\e^2}$ of $H(t,\e)$ at the infinity. 
Note that a function $\psi={}^t(\psi_1,\psi_2)$ is a solution to \eqref{eq:OurEq} if and only if ${}^t(-\overline{\psi_2},\overline{\psi_1})$ is so. This implies that $(J_r^+(t), J_r^-(t))$ and $(J_\ell^+(t), J_\ell^-(t))$ are orthonormal bases on ${\mathbb C}^2$ at each $t\in\R$.
Then we can introduce the scattering matrix $S(\e,h)$ as the change of basis between the pairs of Jost solutions: 
\be\label{def_scattering_M}
\left( J_\ell^+, J_\ell^- \right)=
\left( J_r^+, J_r^- \right)S(\e,h), \quad 
S(\varepsilon,h) =\left(
\begin{array}{cc}
s_{11}(\e,h) & s_{12}(\e,h)\\
s_{21}(\e,h) & s_{22}(\e,h)
\end{array}
\right).
\ee
This matrix is unitary. In particular, one has $|s_{11}|=|s_{22}|$, $|s_{12}|=|s_{21}|$, and $|s_{11}|^2+|s_{21}|^2=1$.  
\begin{definition}\label{def_P}
The transition probability $P(\varepsilon, h)$ is defined by 
$$
P(\varepsilon, h):=|s_{21}(\varepsilon, h)|^2.
$$
\end{definition}

\begin{remark}
    The above definition of the transition probability is equivalent to \eqref{eq:defTP}. In fact, one has $\|J_\bullet^\pm(t)\|_{\C^2}=1$ for any $t$, and 
    \begin{align*}
    &\lim_{t\to-\infty}\left\|\Pi_\pm J_\ell^\pm(t)\right\|_{\C^2}=1,
    &&\lim_{t\to-\infty}\left\|\Pi_\mp J_\ell^\pm(t)\right\|_{\C^2}=0,\\
    &\lim_{t\to+\infty}\left\|\Pi_\pm J_r^\pm(t)\right\|_{\C^2}=1,
    &&\lim_{t\to+\infty}\left\|\Pi_\mp J_r^\pm(t)\right\|_{\C^2}=0.
    \end{align*}
\end{remark}

\begin{condition}\label{condi_2}
The function $V(t)$ has a finite number of zeros $t_1>\cdots>t_n$ on $\mathbb R$, where each zero $t_k$  for $k = 1, \ldots n$ is of finite order denoted by $m_k$. 
\end{condition}

This assumption implies that for $k = 1, \ldots, n$,
\be
V^{(l)}(t_k)=0\quad(1 \leq l<m_k),\quad v_k:=V^{(m_k)}(t_k)\neq0.
\ee
Let $m_*$ denote the maximal order of the zeros:
\be
m_* = \max_{j\in \{1, 2, \ldots, n\}} m_j
\ee
and let $\Lambda_*$ denote the index set of $k \in \{ 1, 2, \ldots, n\}$ which attains $m_*$ (i.e., $m_k=m_*\iff k\in\Lambda_*$). 
Put 
\be
\sigma_k := \sum_{j=1}^k m_j
\ee
for $k = 1,2, \ldots, n$. 
Then $V_r=\lim_{t\to+\infty}V(t)>0$ implies that $\sigma_k$ determines the sign of $V(t)$ on each interval $(t_{k+1},t_k)$, and in particular $\sigma_n$ determines the sign of $V_\ell$, namely $(-1)^{\sigma_k}V(t)>0$ for $t_{k+1}<t<t_k$ and $(-1)^{\sigma_n}V_\ell>0$.

As we mentioned in the introduction, the ratio of $\e$ and (a specific power of) $h$ is crucial. 
We set 
\be
\mu_*:=\mu_{m_*},
\ee
where 
\be
\mu_m=\mu_m(\e,h):= \e h^{-\frac{m}{m+1}}
\ee
for each $m\in\mathbb{N}$. 
We focus on 
the regime $\mu_* \ll1$. 
In the case where there exists at least one avoided crossing generated by a tangential intersection, that is $m_*\geq 2$, we obtain the following result.

\begin{figure}
\centering
\includegraphics[bb=0 0 856 224, width=12cm]{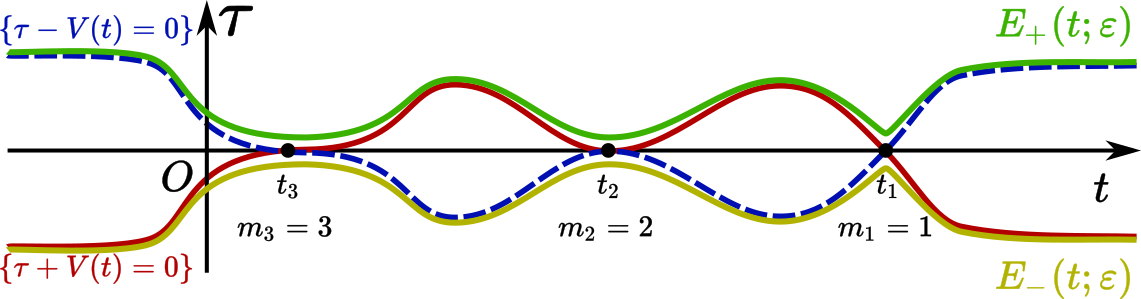}
\caption{An example of $V(t)$ and energies $E_\pm(\e,h)$}
\label{fig2}
\end{figure}



\begin{theorem}\label{mainthm}
Assume Conditions \ref{condi_1}, 
\ref{condi_2} and $m_*\geq 2$. 
Then there exist $\mu_0>0$ and $h_0>0$ such that, 
for any $\e$ and $h$ with $\mu_*(\e,h) \in (0, \mu_0]$ 
and $h \in (0,h_0]$, the transition probability 
$P(\varepsilon, h)$ has the asymptotic expansions:  
$$
P(\varepsilon, h) =\left\{
\begin{array}{lll}
1- C_{*} (h)\, \mu_*^2&+&
\ord\left(\mu_*^2 \left( \mu_* + 
 h^{\frac{1}{m_*(m_*+1)}} \right) \right) 
 \qquad \textrm{if} \ \sigma_n \ \textrm{is odd},\\[12pt]
{\quad } C_{*} (h) \, \mu_*^2&+& 
\ord\left(\mu_*^2 \left( \mu_* + 
 h^{\frac{1}{m_*(m_*+1)}} \right) \right) 
 \qquad \textrm{if} \ \sigma_n \ \textrm{is even},
\end{array}
\right.
$$
where the coefficient $C_{*}(h)$ consists of the product of two factors $\gamma_{*}$ and $\delta_*(h)$, that is $C_{*}(h) = \gamma_{*} \delta_*(h)$, which are given by 
\begin{align}\nonumber
    &\gamma_{*} = 4 \left( \frac{(m_* +1)!}{2} \right)^{\frac{2}{m_* +1}} \!\!
    \varGamma \left(\frac{m_* + 2}{m_*+1} \right)^2 \!\!
    \left( 1 - 
    \frac{1+(-1)^{m_*}}{2} 
    \sin^2 \left(\frac{ \pi}{2(m_* + 1)} \right) \right),\\[10pt]\label{prefactor}
    &\delta_*(h) = \sum_{j\in \Lambda_*} |v_j|^{-\frac{2}{m_*+1}} 
    + 2\!\!\sum_{\substack{j,k \in \Lambda_*\\j<k}}\!\! |v_j v_k|^{-\frac{1}{m_*+1}} \,
    \cos\left(\frac{2}{h}\int_{t_k}^{t_j}V(t)dt+\theta_{m_*}^{j,k}\right),
\end{align}
with
\begin{equation*}
    \theta_{m_*}^{j,k}=
    \left\{
    \begin{aligned}
        &(\ope{sgn}v_j)\frac{\pi}{m_*+1}&&\text{if $m_*$ is odd and }
        \ \ope{sgn}v_j=-\ope{sgn}v_k,\\
        &0&&\text{otherwise}.
    \end{aligned}\right.
\end{equation*}
Here, $\varGamma$ stands for the standard Gamma function $\varGamma(z)=\int_0^{+\infty}t^{z-1}e^{-t}dt$.
\end{theorem}
\begin{remark}
    When every avoided crossing is generated by a transversal intersection, that is, $m_*=1$, Theorem \ref{mainthm} is proven under an additional assumption that $V$ is analytic near the real line \cite{WaZe21_01}. Our method also deduces the same asymptotic formula under Conditions \ref{condi_1}, \ref{condi_2} and the additional condition that $\tilde{\mu}_1:=(\log(1/h))^{{1/2}}\e h^{-1/2}$, replaced with $\mu_1$, is sufficiently small (see also the previous work \cite[Remark 1.2]{Hi23_01}).
\end{remark}
\begin{remark}\label{delta_vanish}
The factor $\gamma_*$ depends only on the highest order $m_*$ of the zeros and never vanishes while the factor $\delta_*(h)$ depends also on the behavior of $V$ not only the local property at zeros and may vanish. This vanishing phenomenon corresponds to the destructive quantum interference.  
Suppose, for example, that $|v_j|$ among $j \in \Lambda_*$ are the same. Put $N_*:= \# \Lambda_*$ and $n_* := \min \Lambda_*$. Then the condition for $\delta_*(h)$ to vanish is given by
\be\label{diamond}
N_* + 2 \left(
\sum_{j\in \Lambda_*\setminus\{n_*\}} \cos {\mathcal V}_j +
\sum_{\substack{j,k \in \Lambda_*\setminus\{n_*\}\\j<k}}
\cos ({\mathcal V}_j - {\mathcal V}_k)
\right) = 0,
\ee
where  
\be
{\mathcal V}_j := \frac{2}{h}\int_{t_{n_*}}^{t_j} V(t)dt 
+  \frac{1-(-1)^{m_*}}{2} 
(\ope{sgn}v_j)\frac{\pi}{2(m_*+1)}.
\ee
The algebraic curve \eqref{diamond} in $(N_*-1)$-variables $\{ {\mathcal V}_j  \}_{j\in \Lambda_*\setminus\{n_*\}}$ appears as so-called {\it Fermi surface} in the context of the discrete Laplacian on the $(N_*-1)$-dimensional diamond lattice, which is a generalization of the hexagonal lattice \cite{AIM16_01}.
\end{remark}

The rest of this subsection is devoted to the concrete expression of the transition probability in Theorem \ref{mainthm} for typical models by means of the following geometric quantity on the (time-energy) phase space. 
For each $k=1,2,\ldots,n-1$, we denote the area enclosed by $V(t)$ and $-V(t)$ between $t_{k+1}$ and $t_k$ by
\begin{equation}\label{areaV}
    \cA_k:=2\int_{t_{k+1}}^{t_k}\left|V(t)\right|dt.
\end{equation}

\begin{figure}
\centering
\includegraphics[bb=0 0 1201.5 205, width=14cm]{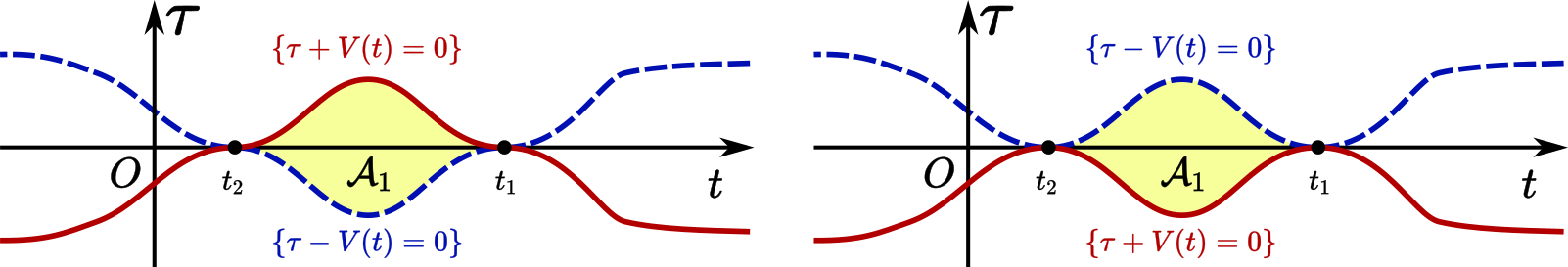}
\caption{Cases (b) (left) and (c) (right) in Example~\ref{2avo} }
\label{fig2-23}
\end{figure}

\begin{example}[Two avoided crossings]\label{2avo} {\ }
Let the number $n$ of avoided crossings be two. Then the transition probability $P(\e,h)$ is $1$ (resp. $0$) modulo $\ord(\mu_*^2)$ if the sum $\sigma_2=m_1+m_2$ of the order of zeros is odd (resp. even). In particular, when the two zeros have the same order, one sees that $P(\e,h)=\ord(\mu_*^2)$ independent of the parity of the order. We give the coefficient $C_{*}(h)$ attached to $\mu_*^2$ in each situation: (Note that $v_1$ is always positive.)
\begin{enumerate}
    \item[{\rm (a)}.] $m_1>m_2$;
    \begin{align}
        C_{*}(h)=\gamma_{m_1}v_1^{-\frac{2}{m_1+1}}. 
    \end{align}
    \item[{\rm (b)}.] $m_1=m_2\in2\Z-1$ and $|v_1| = |v_2|$ $(\iff v_1=-v_2>0)$; 
        \be\label{WZ_BS}
        C_{*}(h) = 4 \gamma_{m_1}  v_1^{-\frac{2}{m_1+1}}
        \cos^2 \left( \frac{\cA_1}{2h}
        - \frac{\pi}{2(m_1+1)}\right).
        \ee
    \item[{\rm (c)}.] $m_1=m_2\in2\Z$ and $|v_1| = |v_2|$ $(\iff v_1=v_2>0)$;
        \be
        C_{*}(h) = 4 \gamma_{m_1}  v_1^{-\frac{2}{m_1+1}} 
        \cos^2 \frac{\cA_1}{2h}.
        \ee
\end{enumerate}
\end{example}
\begin{remark}\label{BS_rem}
In Cases {\rm (b)} and {\rm (c)} of Example \ref{2avo}, we see that $C_*(h)$ may vanish and the order of the transition probability varies due to the destructive quantum interference under the Bohr-Sommerfeld type quantization rule 
\be\label{BS_condition2}
\left\{\begin{aligned}
&\frac{{\mathcal A}_{1}}{h} + \frac{m_1\pi}{m_1+1} \in2\pi\Z &&\text{Case {\rm (b)}},\\
&\frac{{\mathcal A}_{1}}{h} + \pi \in2\pi\Z &&\text{Case {\rm (c)}}.
\end{aligned}
\right.
\ee
This condition is a generalization of that shown in the work \cite{WaZe21_01} (for $m_1=1$).
\end{remark}

\begin{example}[Three avoided crossings]\label{3avo}  {\ }
Let $n=3$.  The transition probability is determined modulo $\ord(\mu_*^2)$ by the sum $(m_1+m_2+m_3)$  whereas the coefficient $C_{*}(h)$ attached to $\mu_*^2$ is determined by zeros $t_j$ only for $j\in\Lambda_*$ and by integrals of $V$ between them. 
In particular, when $\#\Lambda_*\le2$ and $\Lambda_*\neq\{1,3\}$, the coefficient $C_{*}(h)$ is given by the same formula as a model with two avoided crossings.
We remark once again that $v_1$ is always positive.
\begin{enumerate}
    \item[{\rm (a)}.] $\Lambda_*=\{1,3\}$ and $\left|v_1\right|=\left|v_3\right|$ $(\iff$ $v_1=(-1)^{m_1+m_2}v_3>0)$;
    \be
        C_{*}(h)=4\gamma_{m_1}v_1^{-\frac2{m_1+1}}\cos^2\left(\frac{\cA_1+(-1)^{m_2}\cA_2}h\right).
    \ee
    \item[{\rm (b)}.] 
    $m_1=m_2=m_3\in2\Z-1$ and $|v_1| = |v_2| = |v_3|$ $(\iff v_1=-v_2=v_3>0)$;
        \begin{equation}\label{three_example_ood}
        \begin{aligned}
            C_{*}(h) &= \gamma_{m_1} v_1^{-\frac{2}{m_1+1}} 
            \biggl[
            3 + 2\Bigl( \cos \Bigl(\frac{{\mathcal A}_1}{h} - \frac{\pi}{m_1+1}\Bigr) \\
            &\quad\qquad\qquad
            + \cos \Bigl(\frac{{\mathcal A}_2}{h} - \frac{\pi}{m_1+1}\Bigr) 
            + \cos \Bigl(\frac{\cA_1-\cA_2}{h}\Bigr) \Bigr)
            \biggr].
        \end{aligned}
        \end{equation}
    \item[{\rm (c)}.] 
    $m_1=m_2=m_3\in2\Z$ and $|v_1| = |v_2| = |v_3|$ $(\iff v_1=v_2=v_3>0)$;
        \begin{align}
            C_{*}(h) = \gamma_{m_1} v_1^{-\frac{2}{m_1+1}} 
            \left[
            3 + 2\left( \cos \frac{{\mathcal A}_1}{h} 
            + \cos \frac{{\mathcal A}_2}{h}
            + \cos \left(\frac{\cA_1+\cA_2}{h}\right) \right)
            \right].\nonumber
        \end{align}
\end{enumerate}
\end{example}

\begin{remark}
While the destructive quantum interference condition in the case $n=2$ is that the area on the phase space is quantized (i.e. discretized) as in \eqref{BS_condition2}, that condition in $n=3$ is that two areas lie along the Fermi curve. 
\end{remark}

\begin{figure}
\centering
\includegraphics[bb=0 0 856 213, width=10cm]{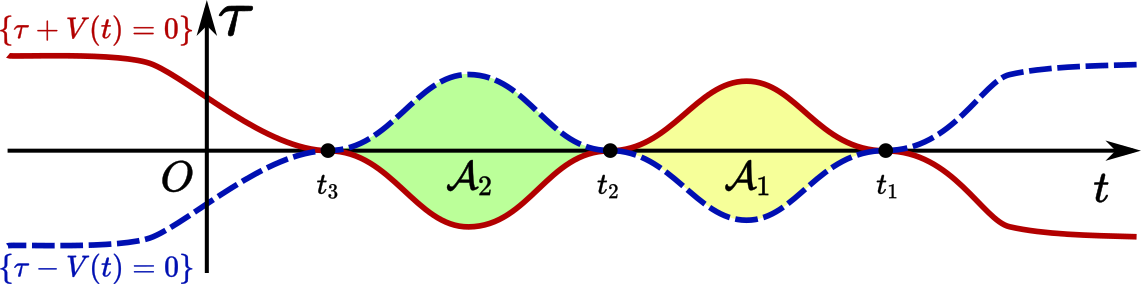}
\includegraphics[bb=0 0 856 213, width=10cm]{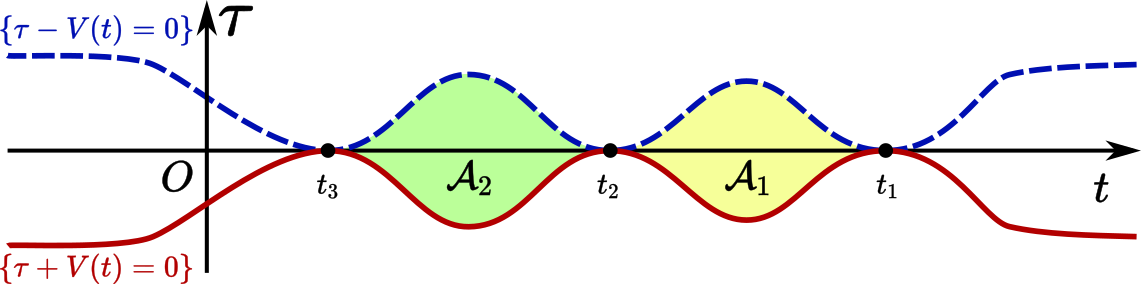}
\caption{Cases b (above) and c (below) in Example~\ref{3avo}}
\label{fig3-3}
\end{figure}

\subsection{Coexistence of the two parameter regimes}\label{coexist}
\begin{figure}
\centering
\includegraphics[bb=0 0 2097 750, width=10cm]{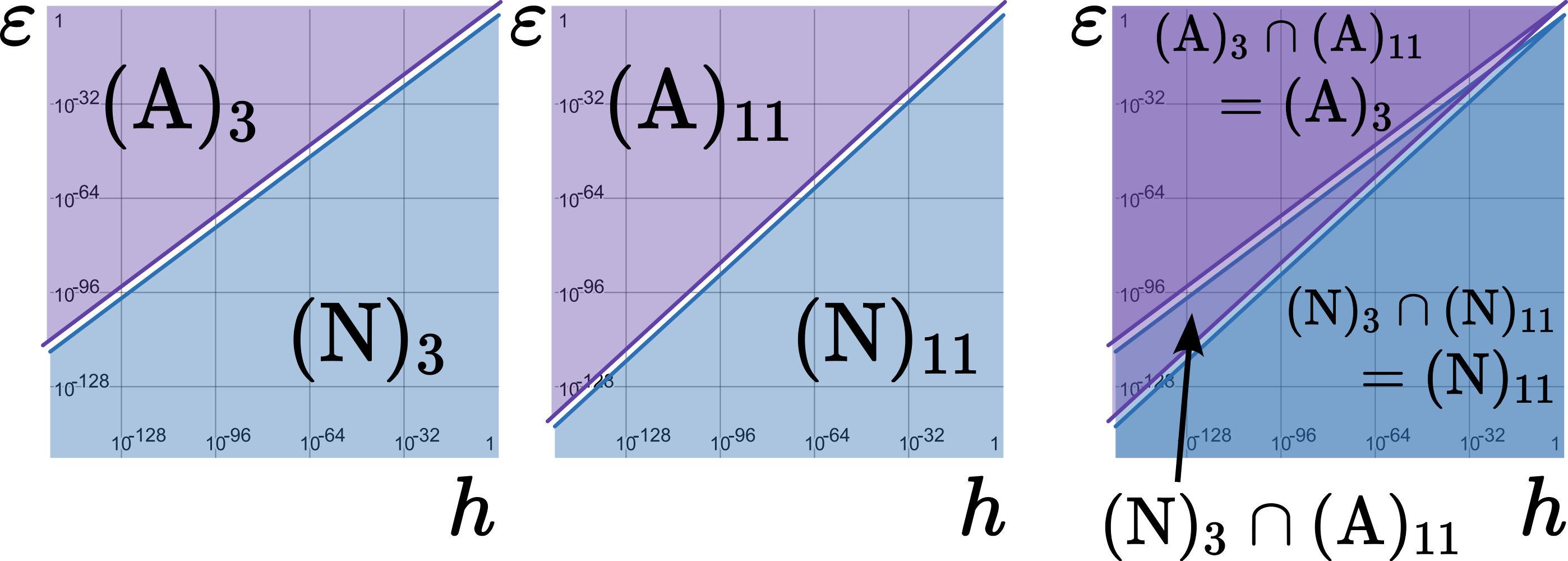}
\caption{Adiabatic and non-adiabatic regimes $({\rm A})_m$ and $({\rm N})_m$ for $m=3,11$ (logarithmic scale, $10^{-150}\le\e,h\le 1$, $({\rm A})_m=\{(\e,h);\,\mu_m\ge100\}$, $({\rm N})_m=\{(\e,h);\,0<\mu_m\le0.01\}$).}
\label{fig4}
\end{figure}

Recall that the quantum dynamics around each avoided crossing near $t=t_k$ depends principally on the magnitude of the parameter $\mu_{m_k}$. More precisely, $\mu_{m_k}\ll1$ and $\mu_{m_k}\gg1$ correspond to the non-adiabatic and adiabatic regimes (note that the regime $\mu_{m_k}\sim1$ is studied \cite{CoLoPo99_01} only for the transversal case $m_k=1$). This parameter is different for two zeros of $V(t)$ with different order, thus the transition problem with several avoided crossings generated by tangential intersections admits various regimes. 

Note that $\mu_m$ obeys the algebraic order relation:
\begin{equation}\label{order_relation-mu}
    m<m' \iff \mu_{m}<\mu_{m'}.
\end{equation}
The regime $\mu_{m_*}\ll1$ considered in Theorem~\ref{mainthm} corresponds to non-adiabatic regime $\mu_{m_k}\ll1$ for every $k\in\{1,\ldots,n\}$.
Conversely, the regime $\mu_{m_\cdt}\gg1$ (with $m_\cdt$ standing for the minimum order $\min_{k\in\{1,\ldots,n\}}m_k$) considered in the previous work \cite{Wa06_01} corresponds to adiabatic regime $\mu_{m_k}\gg1$ for every $k$. 

Here, we consider the case that the two different regimes coexist, that is, the set of indices is decomposed into a disjoint union of two parts
\ben
\{1,2,\ldots,n\}=\overline{\LA}\sqcup\overline{\LN}
\een
such that
\begin{align*}
\mu_{m_k}\gg1\quad &(\forall k\in\overline{\LA}\text{: adiabatic regime}),\\
\mu_{m_k}\ll1\quad &(\forall k\in\overline{\LN}\text{: non-adiabatic regime}).
\end{align*}
Again by \eqref{order_relation-mu}, this corresponds to 
\ben
\muA:=\mu_{\mA}\gg1,\quad \muN:=\mu_{\mN}\ll1,
\een
where we put $\mN:=\max_{k\in\overline{\LN}}m_k$ and $\mA:=\min_{k\in\overline{\LA}}m_k$. We also put
\ben
\LN:=\{k;\,m_k=\mN\}\subset\overline{\LN},\quad \LA:=\{k;\,m_k=\mA\}\subset\overline{\LA}.
\een

Figure~\ref{fig4} illustrates the regimes for $m=3,11$. When each zero of $V$ is either of order 3 or 11, we here study the regime $({\rm N})_3\cap({\rm A})_{11}$ while  Theorem~\ref{mainthm} and the work \cite{Wa06_01} concern the regime $({\rm N})_{11}$ and $({\rm A})_3$, respectively. 
In Figure~\ref{fig5}, the problem here corresponds to $({\rm N})_1\cap({\rm A})_2$ or $({\rm N})_2\cap({\rm A})_3$. Note also that these figures are displayed with a logarithmic scale. Hence the borders between regimes are straight lines. 
Indeed, the border $\mu_m=c$ for some $c>0$ is rewritten as $\log\e=\log c+\frac m{m+1}\log h$.

\begin{figure}
\centering
\includegraphics[bb=0 0 730 613, width=5cm]{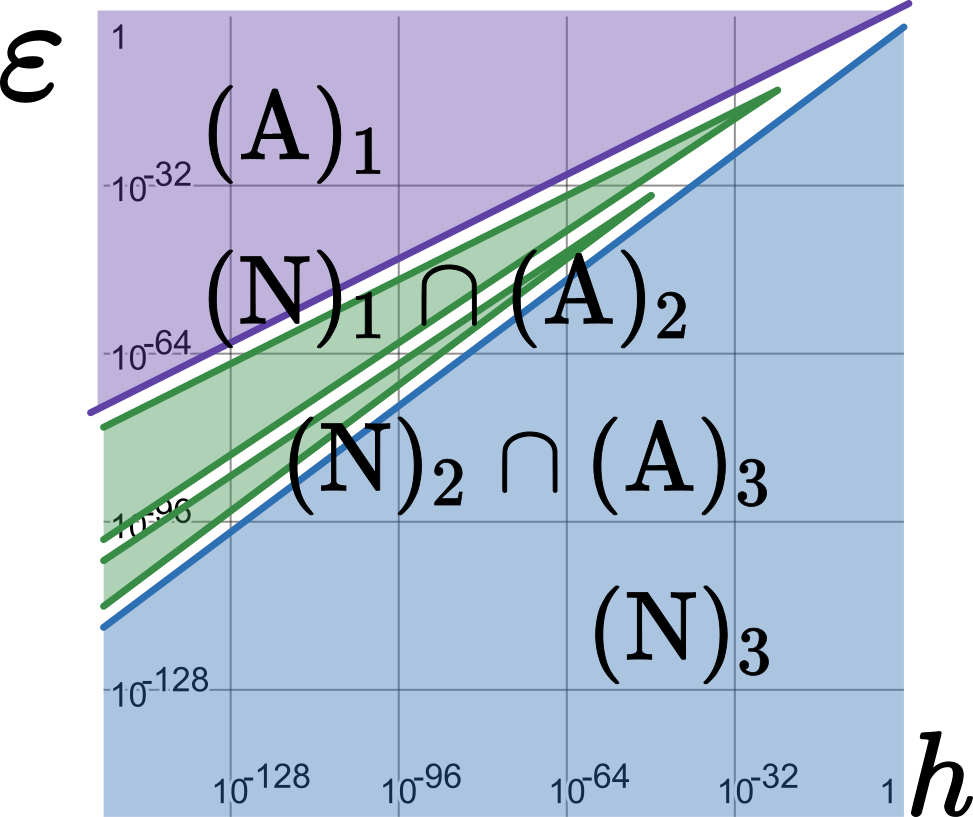}
\caption{Adiabatic and non-adiabatic regimes $(A)_m$ and $(N)_m$ for $m=1,2,3$.
}
\label{fig5}
\end{figure}

In the study of adiabatic regime, one of the authors employed \cite{Wa06_01} the exact-WKB method which requires the function $V$ to be analytic. Hence we also suppose the additional condition.

\begin{condition}\label{condi_3}
$V(t)$ is real-analytic on an interval containing $[t_n, t_1]$.  
\end{condition}

Under this condition, when $\e$ is small enough, there exist $2m_k$  
zeros of $V(t)^2 + \e^2$ near each $t= t_k$ like the power roots. 
We call these zeros turning points and denote the nearest two turning points to the real axis on the upper half-plane by $\zeta_{k,1}(\e)$, $\zeta_{k,m_k}(\e)$, which behave like 
\begin{equation}\label{asym_tp}
    \zeta_{k,j}(\e) \sim t_k + \left( \frac{m_k!}{v_k} \e \right)^{1/m_k}\exp\left[\frac{2j-1}{2m_k} \pi i \right]
\end{equation}
as $\e\to 0$.
We define the action integral $A_{k,j}(\e)$ for $j =1, m_k$  by
\ben
A_{k,j} := 2 \int_{t_k}^{\zeta_{k,j}(\e)} \sqrt{V(t)^2 + \e^2}\, dt,
\een
where the path is the segment from $t_k$ to $\zeta_{k,j}(\e)$ and the branch of the square root of the integrand is $\e$ at $t=t_k$. 
According to the work \cite{Wa06_01}, the asymptotics of the transition probability is given in terms of the turning points in the upper half plane and of the action integrals associated with them.
Moreover, the ratio of the  contribution coming from the other turning points than $\zeta_{k,1}$ and $\zeta_{k,m_k}$ compared with the contribution from $\zeta_{k,1}$ and $\zeta_{k,m_k}$ is exponentially small with respect to $\mu_{m_k}$.
Note that ${\rm Im}\, A_{k,j} >0$ on this branch. The action integral $A_{k,j}(\e)$ behaves like 
\begin{equation}\label{eq:expand-ImA}
{\rm Im}\, A_{k,j} = a_{k} \e^{(m_k + 1)/m_k} + b_{k,j} \e^{(m_k + 2)/m_k}  + \ord\left(\e^{(m_k + 3)/m_k}\right)
\end{equation}
as $\e \to 0$, where $a_k>0$ and $b_{k,1}=-b_{k,m_k}>0$ are constants independent of $\e$.

Roughly speaking, the absolute value of the ``probability amplitude of the transition around an avoided crossing near $t_k$" is small in the limit $\mu_{m_k}\to+\infty$. Contrary to the non-adiabatic case $\LN$, this fact is independent of the parity of $m_k$.
The probability amplitude has the same order as
\ben
\exp\left[-a_k\mu_{m_k}^{(m_k+1)/m_k}\right]\ll1.
\een
From the sake of distinguishing this difference, we introduce
\begin{align*}
\overline{\LA^{\rm odd}} = \{ k\in \overline{\LA}\,;\, m_k:\, \text{odd} \}. 
\end{align*}
Let $n_o=\# \overline{\LA^{\rm odd}}$ be the number of the elements of $\overline{\LA^{\rm odd}}$, and let the elements be labeled in the ascending order $k(1)<k(2)<\cdots<k(n_o)$:
\be\label{def-k(l)}
\overline{\LA^{\rm odd}} = \{k(1), k(2), \ldots, k(n_o) \}.
\ee

We also introduce the effective energy $\tilde V(t)=\tilde V(t;\mN,\mA)$ in this regime by
\be\label{modi_V}
\tilde V (t) = 
(-1)^{l_j}V(t) \quad \text{for $t\in(t_j,t_{j-1})$},
\ee
where $l_j$ is the largest integer such that $k(l_j)\le j$ with the convention $k(0)=0$ 
(see also \eqref{eq:effective-eng}).

Putting $a := \min_{k\in \LA} a_k$ 
and introducing two functions 
\begin{align*}
\epsilon_1=\epsilon_1(\mN, \mA, a) &= \muN + \exp\left[ -a \muA^{(\mA +1)/\mA}
\right],
\\
\epsilon_2=\epsilon_2(\mN,\mA,a) & = \muN \left( \muN + h^{1/{(\mN(\mN+1))}} \right) +\muA^{-(\mA+1)/\mA}\exp\left[-a\muA^{(\mA+1)/\mA}\right],
\end{align*}
we state the asymptotic expansion of the transition probability in this intermediate regime:
\begin{theorem}\label{2ndthm}
Assume Conditions \ref{condi_1}, 
\ref{condi_2} and \ref{condi_3}. 
Then there exist $0<\mu_0<1$ and $h_0>0$ such that, 
for any $\e$ and $h$ satisfying $\muN <\mu_0<\mu_0^{-1}<\muA$, and $h \in (0,h_0]$, the transition probability 
$P(\varepsilon, h)$ has the asymptotic expansions:  
$$
P(\varepsilon, h) =\left\{
\begin{array}{lll}
1- {\mathcal L}(\e,h) &+&
{\mathcal E}(\e,h)
 \qquad \textrm{if} \ (\sigma_n +n_o) \ \textrm{is odd},\\[12pt]
{\quad } {\mathcal L}(\e,h)&+& 
{\mathcal E}(\e,h)
 \qquad \textrm{if} \ (\sigma_n +n_o) \ \textrm{is even},
\end{array}
\right.
$$
where the leading term ${\mathcal L}(\e,h) = \ord(\epsilon_1^2)$ and the error therm ${\mathcal E}(\e,h) = \ord(\epsilon_1 \epsilon_2)$. 
\end{theorem}
\begin{remark}
The parity which characterizes the transition probability depends not only on $\sigma_n$ determined by $V$ but also on $n_o$ determined by the regime. 
This implies that the switch of $P(\e,h)$ occurs with changing the regime without doing the energy $V$ (see Figure~\ref{fig1}).
\end{remark}

As we mentioned in Section \ref{coexist}, Theorem \ref{2ndthm} covers the range of the pair of the parameters $(\e,h)$ included in the parameter regime determined by $\mN$ and $\mA$. 
Regarding $\e$ as a function of $h$ like a one-parameter problem, we find the typical cases, which realize the intermediate regime $\muA \to \infty$ and $\muN \to 0$. 
\begin{description}
    \item[Polynomial case:] If $\e\sim h^\alpha$ with 
    $$\frac{\mN}{\mN+1}<\alpha<\frac{\mA}{\mA+1},$$
    the contribution coming from $\LA$ is exponentially small. 
    \item[Logarithmic case:] If $\e = (h\log (1/h^\rho))^{\mA/(\mA +1)}$ with some positive constant $\rho$,
    the contribution coming from $\LA$ must be taken into account, since $\exp [-a\muA^{(\mA+1)/\mA}] = h^{a\rho}$.
\end{description}

In the former case, the leading term is similar to that in Theorem \ref{mainthm} and is given by
\begin{align}\label{leading_former}
{\mathcal L}(\e,h) &= \muN^2\left(
\sum_{j\in \LN} \gammaN |v_{j+1}|^{-\frac{2}{\mN +1}} 
  + 2\!\!\!\! \sum_{\substack{j,k\in \LN\\j<k}} \!\!\!\! {\rm Re}\, C_{j,k}^{\NN}(\e,h) 
  \cos \left[
  \frac{1}{h}\int_{t_k}^{t_j} \tilde V (t)dt
  \right]\right), 
\end{align}
where the factor $C_{j,k}^{\NN}(\e,h)$ is of $\ord(1)$ and consulted in \eqref{form_flafal}. 
In other cases including the latter case, the leading term is more complicated than \eqref{leading_former}. In fact, the leading term ${\mathcal L}(\e,h)$ is of the form:
\begin{align*}
&\muN^2\left(
\sum_{j\in \LN} \gammaN |v_{j+1}|^{-\frac{2}{\mN +1}} 
  + 2\!\!\!\! \sum_{\substack{j,k\in \LN\\j<k}} \!\!\!\! {\rm Re}\, C_{j,k}^{\NN}(\e,h) 
  \cos \left[
  \frac{1}{h}\int_{t_k}^{t_j} \tilde V (t)dt
  \right]\right)
  \\
  &\quad 
  + \sum_{k\in \LA} \exp \left[
  -2a_k \muA ^{(\mA +1)/\mA} 
  \right]\\
  &\quad + 2\!\!\!\! \sum_{\substack{j\in \LN,k\in \LA\\j<k}} \!\!\!\! {\rm Re}\, C_{j,k}^{\NA}(\e,h) 
  \muN \exp \left[
  -a_k \muA ^{(\mA +1)/\mA} 
  \right]
  \cos \left[
  \frac{1}{h}\int_{t_k}^{t_j} \tilde V (t)dt
  \right]\\
  &\quad + 2\!\!\! \sum_{\substack{j,k\in \LA\\j<k}} \!\!\! {\rm Re}\, C_{j,k}^{\AAa}(\e,h) 
  \exp \left[
  -(a_j+a_k) \muA ^{(\mA +1)/\mA} 
  \right]
  \cos \left[
  \frac{1}{h}\int_{t_k}^{t_j} \tilde V (t)dt
  \right],
\end{align*}
where $C_{j,k}^{\NA}(\e,h)$ and $C_{j,k}^{\AAa}(\e,h)$ are of $\ord(1)$ and referred in \eqref{form_flasha} and \eqref{form_shasha} respectively. 

\begin{remark}
The mixed terms coming from $\epsilon_1^2$ correspond to quantum interference terms referred in Remark \ref{delta_vanish}. 
The phase shift caused by the integral of the energy $V$ changes into the phase shift done by that of the effective energy $\tilde V$ as in Figure \ref{fig1}.
\end{remark}

In our method, we represent the Jost solution $J_\ell^+$ by several bases. As we mentioned in the introduction, each basis is consists of solutions corresponding to an eigenvector associated with $E_\pm(t,\e)$ in each region between two avoided crossings. Consequently, the absolute value of the coefficients gives $\|\Pi_\pm(t;\e) J_\ell^+\|_{\C^2}$. Moreover, one observes from our proof that 
\be\label{eq:effective-eng}
1\sim\|\Pi_{\sigma(t)}(t;\e) J_\ell^+(t)\|_{\C^2}>
\|\Pi_{-\sigma(t)}(t;\e) J_\ell^+(t)\|_{\C^2}\sim0,
\ee 
outside any $(\e,h)$-independent neighborhood of $\{t_1,\ldots,t_n\}$,  
where we put $\sigma(t):=(-1)^{\sigma_n+n_o}\ope{sgn}\tilde{V}(t)$.
In this sense, (the square of) the modulus of the probability amplitude that the energy follows the curve $\tau=E_{\sigma(t)}(t;\e)$ is almost 1. In Figure~\ref{fig1}, we draw this curve in green.  

\begin{example}\label{ex:intermediate} 
Figure~\ref{fig1} shows an example case with $m_k = 2k-1$ ($k=1,2,3$).
\end{example}

\begin{figure}[htb]
\centering
\includegraphics[bb=0 0 1410 1021, width=13cm]{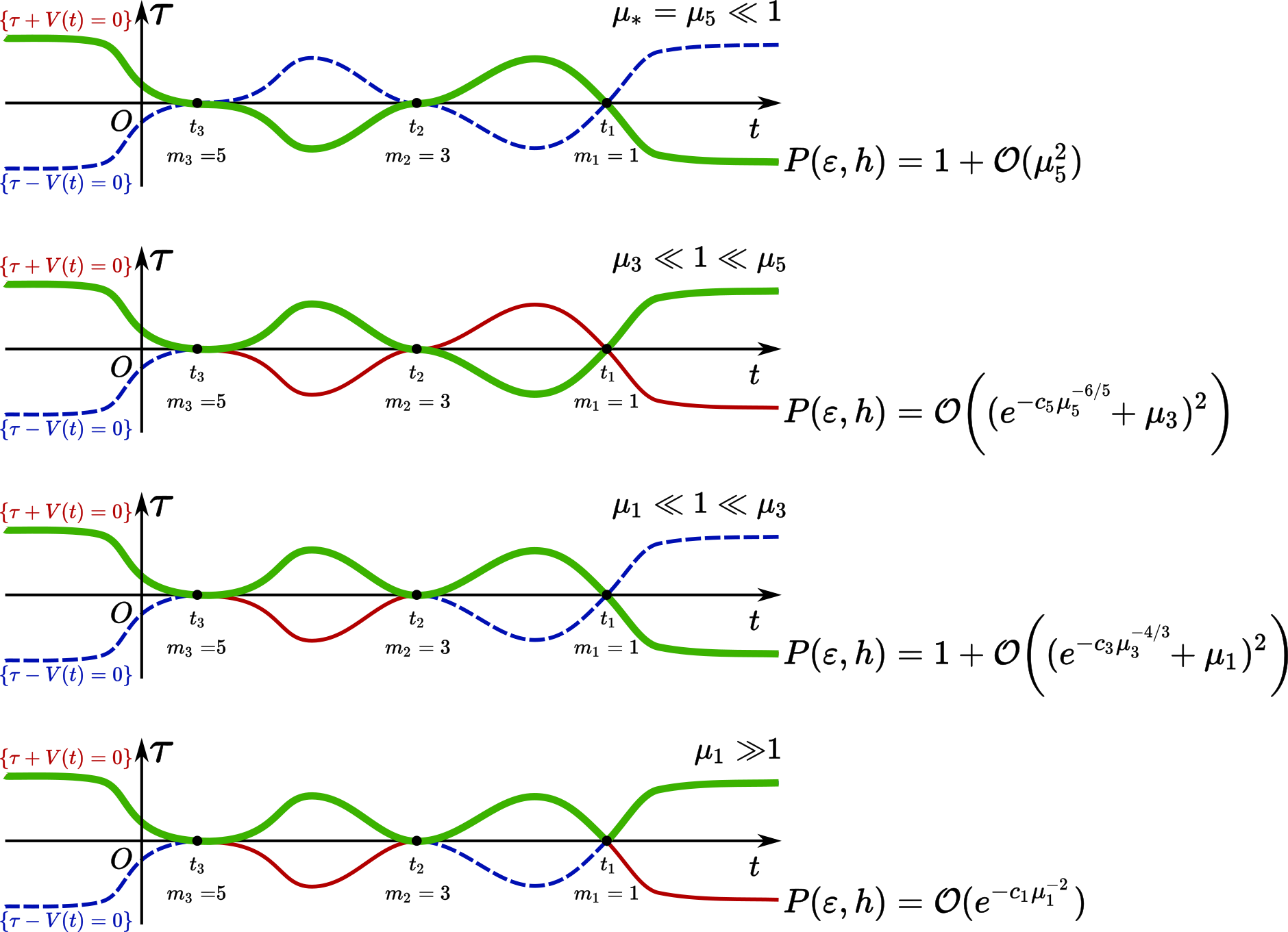}
\caption{Example~\ref{ex:intermediate} for each regime: The green curve illustrates the effective energy $\tau=E_{\sigma(t)}(t;0)=(-1)^{\sigma_n+n_o}\tilde V(t)$ (see also the explanation before Example~\ref{ex:intermediate}).}
\label{fig1}
\end{figure}

\section{Construction of exact solutions}\label{sec_const_sol}

In this section, we construct exact solutions which form a local basis near each vanishing point $t= t_k$ 
by means of a method of successive approximations due to the previous work \cite{FuMaWa19_01}. 
While the equation treated there is a second order $2\times 2$ system of time-independent Schr\"odinger equations, 
our equation in this paper is a first order $2\times 2$ system.  
 
\subsection{Estimates of fundamental solutions}

For simplicity, we assume $t_k=0$, and let $m$ denote $m_k$. Let $I$ be 
an interval that contains the vanishing point $0$ in its interior but no other vanishing points in its neighborhood.
We fix
\ben
u^\pm(t)=\exp\left(\mp\frac ih\int_0^tV(s)ds\right)
\een
as a particular solution to  
\be
\left(hD_t\pm V(t)\right)u=0\quad\text{on }\ I
\ee
respectively, where $D_t$ stands for $-i d/dt$. For any point $a\in I$, we define an integral operator $K^\pm_{a}$ by
\be
K^\pm_{a} f(t):=\frac ih u^\pm (t)\int_{a}^t \frac{f(s)}{u^\pm (s)}ds\quad\text{for }\ f\in C(I),
\ee
where $C(I)$ is the Banach space of continuous functions on $I$ equipped with the norm $\|f\|_{C(I)}:=\sup_{x\in I} |f(x)|$. 
Since 
\be
\left(hD_t \pm V(t)\right)K^\pm_a f=f\quad\text{for }\ f\in C(I), 
\ee
for any $a\in I$, the integral operator $K^\pm_{a}: C(I) \to C(I)$ is well-defined as a fundamental solution of $hD_t \pm V(t)$ respectively for the signs $\pm$.

Using these fundamental solutions $K^\pm_{a^\pm}$ with base points $a^\pm \in I$ respectively, 
our equation \eqref{eq:OurEq} turns into the integral system with arbitrary constants $c^+,c^-\in\C$:
\be\label{reduced_integral_system}
\left\{
\begin{aligned}
&\psi_1(t)=-\e K^+_{a^+}\psi_2(t)
+c^+u^+(t),\\
&\psi_2(t)=-\e K^-_{a^-}\psi_1(t)+c^-u^-(t).
\end{aligned}\right.
\ee
Depending on the choice of the base points $a^+$ and $a^-$, the initial value for $\psi_1$ and $\psi_2$ at these points are determined:
\ben
\psi_1(a^+)=c^+u^+(a^+),\quad
\psi_2(a^-)=c^-u^-(a^-).
\een

In the next subsection, we show a construction of the unique solution to the system by an iteration. For this purpose, we give the following estimates for the fundamental solutions. Note that they are independent of $\e$, and this estimate gives the critical rate $\mu_m=\e h^{-m/(m+1)}$.

Let $\|\cdot\|_q$ for $q\in\R$ be a norm on the space of continuously differentiable functions $C^1(I)$ defined by
\be
\left\|f\right\|_q:=\sup_I|f|+h^q\sup_I|f'|\quad f\in C^1(I).
\ee
\begin{proposition}\label{prop:behavior}
    For any $a^\pm\in I$, there exists $C>0$ such that 
    \be\label{Convergenece_estimate}
\left\|(u^\pm)^{-1}K_{a^\pm}^\pm (u^\mp f)\right\|_{\frac1{m+1}}\le C h^{-\frac m{m+1}}\|f\|_{\frac1{m+1}}
    \ee
    for $h>0$ small enough.
\end{proposition}

For the sake of the proof of Proposition \ref{prop:behavior}, 
we introduce the following lemma which plays an important role in this paper.

\begin{lemma}\label{lem:Int-esti}
On a compact interval $I\subset \R$, consider the integral 
\be
\mc{I}_I(h):=\int_I f(t)\exp\left(\frac {2i}h\int_{0}^t V(s)ds\right)dt,
\ee
with a continuously differentiable function $f\in C^1(I)$ possibly depending on $h$. Then there exists a constant $C>0$ independent of $f$ (but depending on $V$) such that 
\be\label{non-stationary_estimate}
\left|\mc{I}_I(h)\right|\le Ch\sup_I(|f|+|f'|),
\ee
for $h>0$ small enough when $V$ does not vanish on $I$. 
If $0$ is the unique zero in $I$ of $V$, one has
\be\label{dege_stationary_estimate}
\left|\mc{I}_I(h)\right|\le C\left(h^{\frac1{m+1}}\sup_I|f|+h^{\frac2{m+1}}\sup_I|f'|\right)=Ch^{\frac1{m+1}}\|f\|_{\frac1{m+1}},
\ee
where $m$ denotes the order of vanishing at $0$. Moreover if $f$ is independent of $h$, we have
\be\label{dege_stationary_formula}
\mc{I}_I(h)=
f(0)\omega_m h^{\frac1{m+1}}+\ord(h^{\frac2{m+1}}).
\ee
Here, the constant $\omega_m$ is given by
\be\label{eq:omega-m}
\omega_m = 2\left(\frac{(m+1)!}{2|V^{(m)}(0)|}\right)^{\frac1{m+1}}
\varGamma\left(\frac{m+2}{m+1}\right)\, \eta_m,
\ee
with
\be
\eta_m:=\left\{
\begin{aligned}
&\cos\left(\frac{\pi}{2(m+1)}\right)\quad &&m:\text{even},\\
&\exp\left(\frac{\ope{sgn}( V^{(m)}(0))i\pi}{2(m+1)}\right)\quad &&m:\text{odd}.
\end{aligned}\right.
\ee
\end{lemma}

\bigskip

\begin{proof}[Proof of Lemma \ref{lem:Int-esti}]
Suppose that $V$ does not vanish on $I$, that is, a non-stationary case. Then we have for $t\in I$ 
\be\label{eq:IBPoscillation}
\frac h{2iV(t)}\frac{d}{dt}\exp\left(\frac {2i}h\int_{0}^t V(s)ds\right)
=\exp\left(\frac {2i}h\int_{0}^t V(s)ds\right).
\ee
This with an integration by parts and the compactness of $I$ implies the estimate from above of $|\mc{I}_I(h)|$ by $Ch\sup_I(|f|+|f'|)$. 

Suppose next that $0$ is the unique vanishing point of $V$ in $I$. 
Take a smooth cut-off function $\chi$ such that $\chi(t)=1$ for $|t|<Ch^{1/(m+1)}$ and $\chi(t)=0$ for $|t|>2Ch^{1/(m+1)}$. On the support of $1-\chi$, one has the estimate 
\ben
\left|\frac d{dt}\left(\frac {f(t)}{2iV(t)}\right)\right|\le \frac {C_V}{|t|^{m+1}}\left(\sup_I|f|+|t|\sup_I|f'|\right).
\een
Here, the constant $C_V$ depends on $V$ and the support of $\chi$ (independent of $f$).
This estimate with the integration by parts gives the estimate
\begin{align*}
 &\left|\mc{I}_I(h)-\int_I\chi(t)f(t)\exp\left(\frac{2i}h\int_0^t V(s)ds\right)dt\right|\\
= \ &
\left|\int_I(1-\chi(t))f(t)\exp\left(\frac{2i}h\int_0^t V(s)ds\right)dt\right|
\le 
Ch^{\frac1{m+1}}\|f\|_{\frac1{m+1}}.
\end{align*}
On the other hand, the contribution from $\chi$ is estimated by $h^{1/(m+1)}\sup_I|f|$ since the length of the support of $\chi$ is $\ord(h^{1/(m+1)})$, and the estimate \eqref{dege_stationary_estimate} follows.

We then suppose also that $f$ is independent of $h$. Let $g=g(t)$ be the smooth function defined near $0$ such that 
\ben
2\int_{0}^t V(s)\, ds = t^{m+1}g(t),\quad
g(t)=\frac{2V^{(m)}(0)}{(m+1)!}+\ord(t).
\een
Take a smooth cut-off function $\chi$ whose value is 1 near $0$ and supported only on a small neighborhood where the change of the variable $\tau=t|g(t)|^{1/(m+1)}$ is valid. 
Then one has
\ben
\int_I\chi(t)f(t)\exp\left(\frac{2i}h\int_{0}^t V(s)ds\right)dt
=|g(0)|^{-\frac 1{m+1}}
\int_{\R} \chi(t(\tau))\tilde f(\tau)e^{\sigma i\tau^{m+1}/h}\,d\tau
\een
with $\sigma = \ope{sgn} g(0)$ and a smooth function $\tilde f=\tilde f(\tau)$ satisfying $\tilde f(0)=f(0)$. The resulting asymptotic formula is obtained from this integral by applying the method of degenerate stationary phase (see e.g. \cite{Ho83_01}). 
The non-stationary estimate \eqref{non-stationary_estimate} is applicable on the support of $1-\chi$.
\end{proof}

Based on this lemma, let us prove Proposition \ref{prop:behavior}.  
\begin{proof}[Proof of  Proposition \ref{prop:behavior}]
For $f\in C^1(I)$, we have by definition
\ben
(u^\pm)^{-1}K^\pm_{a^\pm}(u^\mp f)(t)
=\frac ih \int_{a^\pm}^t\exp\left(\mp\frac{2i}h\int_0^s V(r)dr\right)f(s)\,ds.
\een
According to \eqref{dege_stationary_estimate} of Lemma~\ref{lem:Int-esti}, 
this integral is estimated by $Ch^{\frac1{m+1}}\|f\|_{\frac1{m+1}}$.
For the derivative, we have
\ben
\frac d{dt}\left[(u^\pm)^{-1}K^\pm_{a^\pm}(u^\mp f)(t)\right]=\frac ih \exp\left(\mp\frac{2i}h\int_0^t V(r)dr\right)f(t).
\een
This is clearly bounded by $h^{-1}\sup_I|f|$.
\end{proof}

\begin{remark}\label{better_estimate}
The argument in the proof of Lemma~\ref{lem:Int-esti} shows that the estimate \eqref{Convergenece_estimate} becomes better as $\ord(h)$  
if the integral interval does not contain $t=0$. 
\end{remark}

\subsection{Method of successive approximations (MSA)}

From Proposition \ref{prop:behavior}, it follows that for each ${}^t (c^+,c^-)\in\C^2$, there uniquely exists a solution to the integral system  \eqref{reduced_integral_system}. By a linearity of the system, the solution is given by the linear combination $c_1w_1(t)+c_2w_2(t)$ of the solutions $w_1(t)$ and $w_2(t)$ corresponding to the choices $\bm{e}_1={}^t(1,0)$ and $\bm{e}_2={}^t(0,1)$ for ${}^t(c^+,c^-)$. Moreover, the solution can be constructed by MSA:
\be\label{def:MSA1}
w_1(t) =w_1(t;a^-,a^+):=
\begin{pmatrix}
\displaystyle\sum_{k\ge0}(\e^2 K^+_{a^+}K^-_{a^-})^k u^+\\
\dip-\e K^-_{a^-}\sum_{k\ge0}(\e^2 K^+_{a^+}K^-_{a^-})^k u^+
\end{pmatrix} 
\ee 
\be\label{def:MSA2}
\left(\text{resp.}\ 
w_2(t) =w_2(t;a^+,a^-):=
\begin{pmatrix}
\dip-\e K^+_{a^+}\sum_{k\ge0}(\e^2 K^-_{a^-}K^+_{a^+})^k u^-\\
\dip\sum_{k\ge0}(\e^2 K^-_{a^-}K^+_{a^+})^k u^-
\end{pmatrix} \right) 
\ee
in $I$ for a fixed small $\e h^{-\frac{m}{m+1}}$.  
These iteration formulas imply that the solutions admit the asymptotic expansions when $\mu_m := \e h^{-\frac{m}{m+1}} \to 0$
as follows:
\begin{equation}\label{asym_vec}
\begin{aligned}
w_1(t; a^-,a^+) &=
\begin{pmatrix}
u^+(t)+\displaystyle \e^2  K^+_{a^+}K^-_{a^-} u^+(t)\\[5pt]
\dip -\e K^-_{a^-} u^+(t)
\end{pmatrix} + 
\begin{pmatrix}
\displaystyle \ord(\mu_m^4)\\[5pt]
\dip \ord(\mu_m^3)
\end{pmatrix},\\[7pt]
w_2(t; a^+,a^-) &=
\begin{pmatrix}
\dip -\e K^+_{a^+} u^-(t)\\[5pt]
u^-(t)+ \displaystyle \e^2  K^-_{a^-}K^+_{a^+} u^-(t)
\end{pmatrix} + 
\begin{pmatrix}
\dip \ord(\mu_m^3)\\[5pt]
\displaystyle \ord(\mu_m^4)
\end{pmatrix}. 
\end{aligned}
\end{equation}
Notice that Proposition \ref{prop:behavior} allows us to choose arbitrarily the base points $a^\pm$ of the fundamental solutions in this construction. 
As we mentioned in Remark \ref{better_estimate}, 
we obtain better asymptotic formulas 
\be\label{eq:behavior+}
w_1(t; a^-,a^+)=\begin{pmatrix}
u^+(t) +\ord(\e^2/h)\\\ord(\e)
\end{pmatrix}\quad\text{on }\ I\cap\{\pm t>0\}
\ee
if $\pm a^->0$, that is, the integral interval $[a^-, t]$ of $K^-_{a^-}$ does not contain the zero, and likewise 
\be\label{eq:behavior-}
w_2(t; a^+,a^-)=\begin{pmatrix}
\ord(\e)\\u^-(t)  +\ord(\e^2/h)
\end{pmatrix}\quad\text{on }\ I\cap\{\pm t>0\}
\ee
if $\pm a^+>0$, that is, the integral interval $[a^+, t]$ of $K^+_{a^+}$ does not. 

\bigskip

Take $(\e,h)$-independent constants $r$ and $\ell$ such that $\ell < 0 < r$ and $[\ell, r] \subset I$.  
We define the four MSA solutions $w_{1.r}, w_{2,r}, w_{1,\ell}$, and $w_{2,\ell}$ in $I$ as
\begin{equation}\label{eq:DefExSol}
\begin{aligned}
&w_{1,r}(t) := w_1(t; r, r), &&w_{2,r}(t) := w_2(t; r, r), 
\\
&
w_{1,\ell}(t) := w_1(t; \ell, \ell), &&w_{2,\ell}(t) := w_2(t; \ell, \ell).
\end{aligned}
\end{equation}

According to Proposition~\ref{prop:behavior} and the asymptotic formula \eqref{asym_vec}, one sees that the asymptotic behaviors of these MSA solutions as $\mu_m \to 0$
are clear on $[\ell,r]$, 
and also that $w_{j,r}(t)$  (resp. $w_{j,\ell}(t)$) $(j=1,2)$ behave like the initial data $u_j(t)$ near $t= r$ (resp. $t=\ell$) for a fixed small $\mu_m$.    
Moreover the pairs $(w_{1,r}(t), w_{2,r}(t))$ and $(w_{1,\ell}(t), w_{2,\ell}(t))$ form bases of the space of solutions.

\section{Connection formulas}\label{sec_connection}

In this section, we also treat a zero of $V(t)$ as $t_k = 0$ for simplicity. 
The purpose of this section is to establish the two connection formulas. One of them connects across the vanishing point and the other does between the consecutive vanishing points.  

\subsection{Across the vanishing point}

The crucial point of this proof is the connection formula between 
the two bases $(w_{1,r}(t),w_{2,r}(t))$ and $(w_{1,\ell}(t),w_{2,\ell}(t))$ introduced by \eqref{eq:DefExSol} in the previous section. 
The claim of this subsection is the asymptotic behavior of the transfer matrix $T(\e,h)$ as follows:
\begin{theorem}\label{thm:Connection}
For the solutions defined by \eqref{eq:DefExSol}, we have 
\be\label{T_bases}
\begin{pmatrix}
    w_{1,\ell}(t) & w_{2,\ell}(t)
\end{pmatrix}=
\begin{pmatrix}
    w_{1,r}(t) & w_{2,r}(t)
\end{pmatrix}
T(\e,h),
\ee
where the $2\times2$-matrix $T=T(\e,h)$ admits the following asymptotic formula:
\be\label{eq:Asym-T}
T(\e,h)=I_2-i\mu_m T_{\ope{sub}}+\ord(\mu_m^2+\mu_m h^{\frac1{m+1}})
\ee
as $(\e,h) \to (0,0)$ with $\mu_m =\e h^{-\frac m{m+1}} \to 0$ and
\be
T_{\ope{sub}}=\begin{pmatrix}0&\omega_m\\\overline{\omega}_m&0\end{pmatrix},
\ee
where $\omega_m$ is given by \eqref{eq:omega-m} with $t_0 = 0$.
\end{theorem}

\begin{remark}
By construction and the symmetry $\overline{K_a^\pm f} = - K_a^\mp\overline{f}$, we have
\be\label{base_symmetry} 
\begin{pmatrix}
0&1\\-1&0
\end{pmatrix}\overline{w_{2,\bullet}} =w_{1,\bullet}\quad(\bullet =r,\ell).
\ee
This makes $T$ symmetrical in the following sense:
\be\label{eq:Unitary-T}
T=\begin{pmatrix}\tau_1&-\tau_2\\\overline{\tau_2}&\overline{\tau_1}\end{pmatrix} \in \su (2),
\ee
where $\su (2)$ is the special unitary group of degree $2$ (see 
Appendix \ref{app_alg_lem} and \eqref{set_M}).  
Namely, $T$ is unitary $(|\tau_1|^2+|\tau_2|^2=1)$. This is a consequence that the time evolution by $H$ is unitary and that for $\bullet=\ell,r$, the basis\ben
(w_{1,\bullet}(\bullet),w_{2,\bullet}(\bullet))=\left(\begin{pmatrix}
    u^+(\bullet)\\0
    \end{pmatrix},
    \begin{pmatrix}
    0           \\u^-(\bullet)
\end{pmatrix}\right)
\een
of $\C^2$ is orthonormal.
\end{remark}


\begin{proof}[Proof of Theorem \ref{thm:Connection}] 
Since the pair $(w_{1,r},\, w_{2,r})$ forms a basis of the space of solutions, the matrix $(w_{1,r}(t)\, w_{2,r}(t))$ is invertible for any $t\in I$. Then \eqref{T_bases} is rewritten as
\ben
T(\e,h) = (w_{1,r}(t)\, w_{2,r}(t))^{-1} (w_{1,\ell}(t)\, w_{2,\ell}(t)),
\een
where the right-hand side is also independent of $t$.
Substituting $r$ for $t$, we have 
\be\label{T_asym_comp_02}
\begin{aligned}
    T(\e,h) &=
    \begin{pmatrix}
        u^-(r)&0\\0&u^+(r)
    \end{pmatrix}
    (w_{1,\ell}(r) \, w_{2,\ell}(r))\\
    &=\mathrm{Id}-
    \begin{pmatrix}
    0 & u^-(r) \e K_{\ell}^+u^-(r)\\
    u^+(r) \e K_{\ell}^- u^+(r) & 0
    \end{pmatrix}
    + \ord(\mu_m^2). 
\end{aligned}
\ee
Here, we have used $u^+u^-=1$ and \eqref{asym_vec}.

In order to prove Theorem \ref{thm:Connection}, it is enough to compute the asymptotic behavior of the quantity $u^-(r) \e K_{\ell}^+u^-(r)$ 
as $(\e,h) \to (0,0)$ with $\mu_m \to 0$. 
The computation of 
\begin{align}\label{T_asym_comp_04}
u^-(r) \e K_{\ell}^+ u^-(r) &= 
\frac{i\e}{h}
\int_{\ell}^r \exp\left(\frac{2i}{h}\int_0^t V(s)ds\right)\,dt
\end{align}
is carried out based on Lemma \ref{lem:Int-esti}. In fact, since the integral interval $[\ell,r]$ of \eqref{T_asym_comp_04} contains the zero of 
$V(t)$, we have   
\begin{align}\nonumber 
u^-(r) \e K_{\ell}^+ u^-(r) &= 
-\frac{i\e}{h}
\left( 
\omega_m h^{\frac{1}{m+1}} + \ord(h^{\frac{2}{m+1}})
\right)\\ \label{T_asym_comp_05}
&= -i \mu_m \omega_m + \ord(\mu_m h^{\frac{1}{m+1}}).
\end{align}
Combining \eqref{T_asym_comp_02} and \eqref{T_asym_comp_05}, we obtain Theorem \ref{thm:Connection}. 

\end{proof}

Hence this theorem implies that the transfer matrix $T_k$ in \eqref{formula_Smatrix} admits the asymptotics
\be\label{formula_Tk}
T_k(\e,h)= 
\begin{pmatrix}
1&-i \omega_{m_k} \mu_{m_k}\\
-i \overline{\omega}_{m_k} \mu_{m_k}&1
\end{pmatrix}
+\ord(\mu_{m_k}^2+\mu_{m_k} h^{\frac1{m_{k}+1}})
\ee
as $(\e,h) \to (0,0)$ with $\mu_{m_k} =\e h^{-\frac{m_k}{m_k+1}} \to 0$.

\subsection{Between the vanishing points}

In addition to Theorem \ref{thm:Connection}, which gives the connection formula around the vanishing point of $V(t)$, a similar argument yields the following proposition, which gives the connection formula between two consecutive zeros of $V(t)$. 

Let $t_{k+1}<t_k$ be two consecutive zeros of $V(t)$ (i.e. $V\neq0$ on $]t_{k+1},t_k[$) with multiplicities $m_k, m_{k+1}$, and 
let $\ell_j, r_j$ $(j=k,k+1)$ be base points such that 
$\ell_{k+1} < t_{k+1} < r_{k+1} < \ell_k < t_k < r_k$. The configuration of these points implies that each interval $I_j := [\ell_j, r_j]$ includes $t_j$ and does not intersect with each other. 
We set $m_* = \max \{ m_k, m_{k+1}\}$. 
We can consider two bases $(w_{1,\ell_k}, w_{2,\ell_k})$ and $(w_{1,r_{k+1}}, w_{2, r_{k+1}})$, which are similarly given by 
the formulas \eqref{eq:DefExSol} with $t_0 = t_k$ and $t_0 = t_{k+1}$ respectively.

\begin{proposition}\label{prop_Tkk1}
The change of basis $T_{k,k+1}(\e,h)$ given by 
\be\label{asym_Tkk1_01}
 (w_{1,r_{k+1}}, w_{2, r_{k+1}}) = (w_{1,\ell_k}, w_{2,\ell_k}) T_{k,k+1}(\e,h)
\ee
admits the following asymptotic behavior as $(\e,h) \to (0,0)$ with $\mu_{m_*} \to 0$:
\be\label{asym_Tkk1}
T_{k,k+1}(\e,h) = 
\begin{pmatrix}
\exp\left(-\frac{i}{h} \int_{t_{k+1}}^{t_k} V(t)dt \right) & 0\\
0& \exp\left(\frac{i}{h} \int_{t_{k+1}}^{t_k} V(t)dt\right)
\end{pmatrix}
+ \ord\left(\frac{\e^2}{h}\right).
\ee
\end{proposition}
\begin{remark}
The error term in \eqref{asym_Tkk1} is rewritten as 
$\ord(\e^2/h) = \ord(\mu_1^2)$. 
From the order relation \eqref{order_relation-mu} of $\mu_m$ with respect to $m$, this error is smaller than that in \eqref{eq:Asym-T} when $m>1$.
\end{remark}

\begin{proof}[Proof of Proposition \ref{prop_Tkk1}] 
We can derive from \eqref{asym_Tkk1_01} the expression of $T_{k,k+1}$ by a similar way to the proof of Theorem \ref{thm:Connection}. 
We set $u_k^\pm = \exp\left(\mp i\int_{t_k}^t V(s) ds/h\right)$, and compute the matrix $T_{k,k+1}$ by using the value at $t=\ell_k$:
\begin{align*}
  T_{k,k+1}(\e,h)
 &= (w_{1,\ell_k} \  w_{2, \ell_k})^{-1} (w_{1,r_{k+1}}\, w_{2,r_{k+1}}) \bigr|_{t=\ell_k}\\
 &= \begin{pmatrix}
     u^-_k(\ell_k) & 0\\ 0 & u^+_k(\ell_k)
 \end{pmatrix}
 \left(
 \begin{pmatrix}
     u^+_{k+1}(\ell_k) & 0\\ 0 & u^-_{k+1}(\ell_k)
 \end{pmatrix} + \ord\left(\frac{\e^2}h\right)
 \right). 
\end{align*}
Here, we have used the asymptotic formulas \eqref{eq:behavior+} and \eqref{eq:behavior-}. Note that the integral interval $[r_{k+1}, \ell_k]$ of the fundamental solutions for the construction does not contain any zero of $V(t)$. 
We deduce \eqref{asym_Tkk1} from this with the identity
\be
u^\pm_{k}(\ell_k) u^\mp_{k+1}(\ell_k) = \exp\left(\pm\frac{i}{h} \int_{t_{k+1}}^{t_k} V(t) dt\right). 
\ee
\end{proof}

\section{End of the proofs}\label{sec_end_proof}

By using the transfer matrices $T_k$ $(k=1,2,\ldots,n)$, $T_{k,k+1}$ $(k=1,2,\ldots,n-1)$ introduced in the previous section and $T_r$, $T_\ell$ in Appendix~\ref{app_jost}, the scattering matrix $S=S(\e,h)$ is represented as
\be\label{formula_Smatrix}
S = T_r^{-1} T_1 T_{1,2} T_2 T_{2,3} T_3 \cdots T_{n-1,n} T_n T_\ell.
\ee
Here, the matrix $T_r$ (resp. $T_\ell$) has a similar form to $T_{k,k+1}$ which connects Jost solutions $J_r^\pm$ (resp. $J_\ell^\pm$) and local solutions near $t_1$ (resp. $t_n$). 
The previous section shows the asymptotic behaviors of $T_k$ and $T_{k,k+1}$ (see \eqref{formula_Tk}, \eqref{asym_Tkk1}), and 
the appendix does those of $T_r$ and $T_\ell$ (see \eqref{fomula_Tr_02}, \eqref{fomula_Tell_02}, \eqref{fomula_Tell_03}).  
The formula of $T_\ell$ depending on the sign of $V_\ell$ 
can be rewritten as 
\be
T_\ell = \left\{
\begin{aligned}
&T_{n,n+1}  &&\qquad  (\text{$\sigma_n$ is even})\\
&T_{n,n+1} J &&\qquad  (\text{$\sigma_n$ is odd}),
\end{aligned}
\right.
\ee
by means of $\sigma_n = \sum_{k=1}^n m_k$ and 
\be\label{def_J}
T_{n,n+1} := \begin{pmatrix}
    \exp\left(-\frac{i}{h} R_\ell\right) & 0\\
    0 & \exp\left(+\frac{i}{h} R_\ell\right)
\end{pmatrix},\qquad
J :=
\begin{pmatrix}
0 & -1\\ 1 &0
\end{pmatrix},
\ee
where  $J$ is a complex structure on $\C^2$, that is, $J^2 = -\mathrm{Id}$. 
Note that all of these transfer matrices are elements of $\su (2)$, which is mentioned in Appendix~\ref{app_alg_lem}. 
By using the notations ${\mathcal T}$ and $\tau_{jk}$ given in \eqref{notation_TT}, the scattering matrix $S$ is expressed as
$S = T_r^{-1} {\mathcal T}$ (resp. $S = T_r^{-1} {\mathcal T} J$) if $\sigma_n$ is even (resp. odd). This implies that, when $\sigma_n$ is even (resp. odd), the off-diagonal entry $s_{21}$ equals $e^{-iR_r/h}\tau_{21}$ (resp. $e^{-iR_r/h}\tau_{22}$). 
From Lemma \ref{lem_alg_TT} alone, it is complicated to examine the asymptotics of $\tau_{22}$ up to the coefficient of $\ord(\mu^2)$. However, thanks to the unitarity of the scattering matrix $S$, that is, $ |s_{11}|^2 + |s_{21}|^2 = 1$, the computation of the asymptotic behavior of $|s_{21}|^2$ can be reduced to that of $|\tau_{21}|^2$ even if $\sigma_n$ is odd. 

\subsection{Proof of Theorem \ref{mainthm}} 

Let us demonstrate the proof of Theorem \ref{mainthm}. 
Taking the complex numbers $\alpha_k$, $\beta_k(\mu)$ and $\upsilon_k$ in Lemma \ref{lem_alg_TT} as 
\begin{align*}
\alpha_k = 1,\qquad \beta_k(\mu) = -i \overline{\omega_{m_k}} \mu_{m_k},\qquad 
\upsilon_k = \exp\left(-\frac{i}{h} \int_{t_{k+1}}^{t_k} V(t)dt \right), 
\end{align*}
and noting $\mu_{m_k} \ll \mu_*$ for any $k$, we have from the algebraic formula \eqref{t21^2_ver2} the asymptotic behavior of $|\tau_{21}|^2$ as follows: 
\begin{equation}\label{error_compare}
\begin{aligned}
|\tau_{21}|^2 &= \mu_*^2 \sum_{j \in \Lambda_*} |\omega_{m_j}|^2  + 2 \mu_*^2\, \re\, \sum_{\substack{j,k \in \Lambda_*\\j<k}} \overline{\omega_{m_j}} \omega_{m_k} 
 e^{-\frac{2i}{h}\int_{t_k}^{t_j} V(t)dt }
 + \ord(\mu_*^3) + \ord(\mu_* \mu_{*-1}),
\end{aligned}
\end{equation}
where $\mu_{*-1}$ stands for $\mu_{m_* -1}$. 
Concerning the error terms in \eqref{error_compare}, the former one, i.e. $\ord(\mu_*^3)$, is a higher order error term coming from $\Lambda_*$ and the latter is a cross term between the largest vanishing order and the second largest one. This error coming from a cross term is at most $\ord(\mu_* \mu_{*-1})$. 
If $m_*$ is odd, $\omega_{m_j}$ is not real for $j \in \Lambda_*$, and the following phase shift term may arise from the product $\overline{\omega_{m_j}} \omega_{m_k}$ depending on the sign of $v_j$ and $v_k$:
\be
\ope{arg}(\overline{\omega_{m_j}} \omega_{m_k})
=\frac{((\ope{sgn}v_k)-(\ope{sgn}v_j)) \pi}{2(m_*+1)}
=\left\{
\begin{aligned}
-&\frac{(\ope{sgn}v_j) \pi }{2(m_*+1)}   && \text{if}\ \ope{sgn}v_j = -\ope{sgn}v_k,\\
0& && \text{otherwise}.
\end{aligned}
\right. 
\ee
Therefore, when $m_*$ is odd, the quantity $|\tau_{21}|^2$ behaves like 
\begin{align}\nonumber
 &\mu_*^2 (\omega_*^o)^2 
\left(
 \sum_{j \in \Lambda_*} |v_j|^{-\frac{2}{m_*+1}}  + 2 \re\, \sum_{\substack{j,k \in \Lambda_*\\j<k}}  |v_j v_k|^{-\frac{1}{m_*+1}} \cos \left(\frac{2}{h}  
 \int_{t_k}^{t_j} V(t)dt + \theta^{j,k} \right) 
\right)\\ 
\label{formula_t21_odd}
&\quad 
+ \ord(\mu_*^3) + \ord(\mu_* \mu_{*-1}),
\end{align} 
where 
\begin{align*}
\omega_*^o &= 2 \left( \frac{(m_* +1)!}{2} \right)^{\frac{1}{m_* +1}} \!\!
    \varGamma \left(\frac{m_* + 2}{m_*+1}\right),
    \\[7pt] 
\theta^{j,k} &= \left\{
\begin{aligned}
&(\ope{sgn}v_j) \frac{\pi }{2(m_*+1)}  && \text{if}\ \ope{sgn}v_j = -\ope{sgn}v_k,\\
&0 && \text{otherwise.}
\end{aligned}
\right. 
\end{align*}

On the other hand, $\omega_{m_j}$ is real when $m_*$ is even.
In this case, the quantity $|\tau_{21}|^2$ can be computed similarly as the formula \eqref{formula_t21_odd} 
by replacing $\omega_*^o$ with $\omega_*^o \cos(\pi/2(m_*+1))$ and $\theta^{j,k}$ with $0$. 
Therefore we have completed the proof of Theorem \ref{mainthm}.

\hfill $\Box$

\subsection{Proof of Theorem \ref{2ndthm}}

In the intermediate regime where $\muN\to0$ and $\muA\to\infty$, the asymptotic formula \eqref{formula_Tk} of the transfer matrix $T_k$ is valid for $k\in \overline{\LN}$ but not for $k\in\overline{\LA}$.
Note that, the parameter $\mu_{m_k}$ is no longer a small parameter for $k\in\overline{\LA}$ since $\mu_{m_k}\ge\muA$ due to the algebraic order relation \eqref{order_relation-mu} and the definition of $\muA$. 
To prove Theorem~\ref{2ndthm}, we first review the asymptotic behavior of $T_k$ as $\mu_{m_k}\to\infty$ obtained in the work \cite{Wa12_01} (\S\ref{sec:Tk-adia}).
We will find that the transfer matrix $T_k$ is asymptotic to an anti-diagonal matrix in the limit $\muA\to\infty$ if $k\in\overline{\LA^{\rm odd}}$.
We then use a change of bases of the solutions to apply the method of Appendix~\ref{app_alg_lem} for the computation of the product of the transfer matrices (\S\ref{sec:prod-Tk}).
One can see the role of the effective energy $\tilde V$ there.

\subsubsection{Asymptotic behavior of the transfer matrix $T_k$ in the adiabatic regime}\label{sec:Tk-adia}
Let $k\in\overline{\LA}$.  
Strictly speaking, the MSA solutions $w_{j,r_k}$ and $w_{j,\ell_k}$ $(j=1,2)$ are not defined in exactly the same way as in the case $k\in\overline{\LN}$ $(\Leftrightarrow \mu_{m_k}\ll1)$, since the convergence of the infinite series in their definitions (see \eqref{def:MSA1}, \eqref{def:MSA2}, \eqref{eq:DefExSol}) requires the smallness of the parameter $\mu_{m_k}$.  
Fortunately, thanks to Remark~\ref{better_estimate}, under the weaker condition $\mu_1\ll1$, we can still construct exact solutions on each closed, $(\e,h)$-independent interval without vanishing points of $V$, using MSA.  
More precisely, we define $w_{1,r_k}:=w_1(t;r_k,r_k)$ by regarding $K_r^\pm$ as a bounded operator on $(C^1(I_{r_k}),\|\cdot\|_{1/2})$ with $\inf I_{r_k}>t_k$.  
By abuse of notation, we continue to denote these solutions by the same symbols.  
Note that although these solutions are constructed only away from the vanishing point, they can be uniquely continued across it.  
Thus, the transfer matrix $T_k$ can be defined in the same way as in the non-adiabatic case.\\

Before considering $T_k$, that is, the transfer matrix between the bases $(w_{1,r_k},w_{2,r_k})$ and $(w_{1,\ell_k}, w_{2,\ell_k})$, we first study the transfer matrix $T_k^{\rm w}$ between the exact WKB solutions $(\psi_{k,r}^-, \psi_{k,r}^+)$ and $(\psi_{k,\ell}^-, \psi_{k,\ell}^+)$ defined below.
After that, we study two other transfer matrices. One is between $(w_{1,r_k},w_{2,r_k})$ and $(\psi_{k,r}^-, \psi_{k,r}^+)$, and the other one is between $(\psi_{k,\ell}^-, \psi_{k,\ell}^+)$ and $(w_{1,\ell_k}, w_{2,\ell_k})$.
The asymptotics of $T_k$ is governed by the product of these three transfer matrices.

Let us first define the exact WKB solutions. 
From Condition \ref{condi_3} and the asymptotic formula of the turning point \eqref{asym_tp}, we can find for small $\e$ a simply connected domain $\mc{S}_k(\e)\subset\C$ such that the following conditions are satisfied:
\begin{itemize}
    \item $\mc{S}_k(\e)$ is a complex neighborhood of $t_k$.
    \item $\mc{S}_k(\e)$ contains only four turning points $\zeta_{k,1}(\e), \zeta_{k,m_k}(\e), \overline{\zeta_{k,1}(\e)}, \overline{\zeta_{k,m_k}(\e)}$.
    \item $V$ is analytic in $\mc{S}_k(\e)$.
\end{itemize}
The base points of the symbol $\hat r_k, \hat  \ell_k \in \mc{S}_k(\e)$ are chosen with $\re\, \hat \ell_k < t_k < \re\, \hat r_k$ and $\im\, \hat r_k, \im\, \hat \ell_k >0$. 
The exact WKB solutions $\psi_{k,\bullet}^\pm$ $(\bullet= \ell, r)$ in $\mc{S}_k(\e)$  are defined by
\begin{align*}
    \psi_{k,r}^-(t) &:= -ie^{-i \overline{A_{k,1}}/2h}\psi^-(t, \overline{\zeta_{k,1}(\e)}, \overline{\hat r_k};h),\\
    \psi_{k,r}^+(t) &:=    -e^{iA_{k,1}/2h}\psi^+(t, \zeta_{k,1}(\e), \hat r_k;h),\\
    \psi_{k,\ell}^-(t) &:= -ie^{-i \overline{A_{k,m_k}}/2h}\psi^-(t, \overline{\zeta_{k,m_k}(\e)}, \overline{\hat \ell_k};h),\\ 
    \psi_{k,\ell}^+(t) &:= -e^{iA_{k,m_k}/2h}\psi^+(t, \zeta_{k,m_k}(\e), \hat \ell_k;h),
\end{align*}
with the notation of \eqref{def_exactWKBsol}.
Note that the branch appearing in this definition is also taken as in Appendix \ref{Review_WKB}.
As in \cite[Theorem 4.1.1]{Wa12_01}, we obtain the transfer matrix $T_k^{\rm w}$ between the exact WKB solutions $(\psi_{k,\ell}^-, \psi_{k,\ell}^+)$ and $(\psi_{k,r}^-, \psi_{k,r}^+)$ (i.e., $(\psi_{k,\ell}^-, \psi_{k,\ell}^+)=(\psi_{k,r}^-, \psi_{k,r}^+)T_k^{\rm w}$) by employing Lemma \ref{wkblem2} and Lemma \ref{wkblem3}:
The matrix $T_k^{\rm w}$ has the form
$$
T_k^{\rm w} =
\begin{pmatrix}
\alpha_k^{\rm w} & - \overline{\beta_k^{\rm w}}\\
\beta_k^{\rm w} & \overline{\alpha_k^{\rm w}}
\end{pmatrix} \in \su (2),
$$
where $\alpha_k^{\rm w},$ $\beta_k^{\rm w}$ admit the asymptotic expansions
\begin{align}\label{adi_alpha}
\alpha_k^{\rm w} &= 1 + \ord\left(\mu_{m_k}^{-\frac{m_k +1}{m_k}} \right),\\[7pt] \nonumber
\beta_k^{\rm w} &= (-1)^{\sigma_{k-1}}  \left( \exp\left[-\frac{i}{h} \overline{A_{k,1}}\right]  - (-1)^{m_k} \exp\left[-\frac{i}{h} \overline{A_{k,m_k}}\right] \right)\\ \label{adi_beta}
 \qquad &+ \ord\left(\mu_{m_k}^{-\frac{m_k +1}{m_k}} 
\exp\left[-a_k\mu_{m_k}^{(m_k+1)/m_k}\right]\right),
\end{align}
as $\mu_{m_k} \to \infty$ with the notations given in \S \ref{coexist}. 

The asymptotic behavior of the two transfer matrices between the bases consists of the MSA solutions and the exact WKB solutions is computed by comparing the asymptotic behaviors away from the vanishing point $t_k$ of these solutions.
The asymptotic behaviors of the MSA solutions are guaranteed by \eqref{eq:behavior+} and \eqref{eq:behavior-}, while those of the exact WKB solutions are given as follows:
There exist intervals $I_r\subset (t_k,t_{k-1})$ and $I_\ell\subset(t_{k+1},t_k)$ (with the convention $t_0=+\infty$ and $t_{n+1}=-\infty$) and $h$-independent $\C^2$-valued functions $f_{k,r}^\pm$, $f_{k,\ell}^\pm$ such that
\begin{align*} 
    \psi_{k,r}^-(t) &=\exp\left(-\frac ih \int_{t_k}^t\sqrt{V(s)^2+\e^2}\,ds\right)
    \left(f_{k,r}^-(t) (1+\ord(h))\right),\\
    \psi_{k,r}^+(t) &=\exp\left(\frac ih \int_{t_k}^t\sqrt{V(s)^2+\e^2}\,ds\right)
    \left(f_{k,r}^+(t) (1+\ord(h))\right),
\end{align*}
holds as $h\to 0$ uniformly on $I_r$ and that
\begin{align*}
    \psi_{k,\ell}^-(t) &=\exp\left(-\frac ih \int_{t_k}^t\sqrt{V(s)^2+\e^2}\,ds\right)
    \left(f_{k,\ell}^-(t) (1+\ord(h))\right),\\
    \psi_{k,\ell}^+(t) &=\exp\left(\frac ih \int_{t_k}^t\sqrt{V(s)^2+\e^2}\,ds\right)
    \left(f_{k,\ell}^+(t) (1+\ord(h))\right),
\end{align*}
holds as $h\to 0$ uniformly on $I_\ell$. These asymptotics are due to Lemma~\ref{wkblem2}.
For each point of the interval $I_r$ (resp. $I_\ell$), there exist a canonical curve of type $+$ from $\hat {r}_k$ (resp. $\hat\ell_k$) and a canonical curve of type $-$ from $\overline{\hat{r}_k}$ (resp. $\overline{\hat\ell_k}$).
Moreover, as $\e \to 0$ with $\mu_1 = \e^2/h \ll 1$, the behaviors of the phase factors and the functions $f_{k,r}^\pm$ (resp. $f_{k,\ell}^\pm$) can be computed depending on whether $V(t)$ on $I_r$ (resp. $I_\ell$) is positive or not as follows:
\begin{equation}\label{wkb_asym_away}
\begin{aligned}
    \psi_{k,r}^-(t) &= \exp\left[-(-1)^{\sigma_{k-1}} \frac{i}{h} \int_{t_k}^t V(s)ds\right] Q^{\sigma_{k-1}} \begin{pmatrix}
   1\\\ord (\e) 
\end{pmatrix} (1+\ord(\mu_1) + \ord(h)),\\
    \psi_{k,r}^+(t) &=  \exp\left[(-1)^{\sigma_{k-1}} \frac{i}{h} \int_{t_k}^t V(s)ds\right] (-Q)^{\sigma_{k-1}} \begin{pmatrix}
   \ord (\e)\\1 
\end{pmatrix}(1+\ord(\mu_1) + \ord(h)),\\
    \Biggl(\text{resp.}\ 
    \psi_{k,\ell}^-(t) &= \exp\left[-(-1)^{\sigma_{k}} \frac{i}{h} \int_{t_k}^t V(s)ds\right] Q^{\sigma_{k}} \begin{pmatrix}
   1\\\ord (\e) 
\end{pmatrix} (1+\ord(\mu_1) + \ord(h)),\\
    \psi_{k,\ell}^+(t) &= \exp\left[(-1)^{\sigma_{k}} \frac{i}{h} \int_{t_k}^t V(s)ds\right] (-Q)^{\sigma_{k}} \begin{pmatrix}
   \ord (\e)\\1 
\end{pmatrix} (1+\ord(\mu_1) + \ord(h)),\quad \Biggr)
\end{aligned}
\end{equation}
as $(\e,h) \to (0,0)$ with $\mu_1 \ll 1$ uniformly on $I_r$ (resp. $I_\ell$). 
Here $Q$ stands for the matrix $Q= \begin{pmatrix}
    0&1\\1&0
\end{pmatrix}$ and it is employed  that the integrand in the phase factor 
\begin{equation*}
    \sqrt{V(t)^2+\e^2}=(\ope{sgn} V(t))V(t)+\ord(\e^2)\qquad (\e \to 0),
\end{equation*}
away from the vanishing point $t_k$. 
The matrix $Q$ has the useful properties:
\begin{equation}\label{commute_prop}
Q \begin{pmatrix}
    a&b\\c&d
\end{pmatrix}
= \begin{pmatrix}
    d&c\\b&a
\end{pmatrix}Q,\qquad Q^2 = {\rm Id}.
\end{equation}
Notice that the signs of $V(t)$ in the interval $(t_k,t_{k-1})$ and $(t_{k+1},t_k)$ are given respectively by $(-1)^{\sigma_{k-1}}$ and by $(-1)^{\sigma_k}$.

Then one obtains $(w_{1,\bullet_k},w_{2,\bullet_k})=(\psi_{k,\bullet}^-,\psi_{k,\bullet}^+)T_{k,\bullet}$ for $\bullet=r,\ell$,
\begin{equation*}
    T_{k,r}=\left\{
    \begin{aligned}
        &I_2+\ord(\mu_1+h)&&(\sigma_{k-1}:\text{even}),\\
        &J+\ord(\mu_1+h)&&(\sigma_{k-1}:\text{odd}),
    \end{aligned}\right.\quad
    T_{k,\ell}=\left\{
    \begin{aligned}
        &I_2+\ord(\mu_1+h)&&(\sigma_k:\text{even}),\\
        &J+\ord(\mu_1+h)&&(\sigma_k:\text{odd}),
    \end{aligned}\right.
\end{equation*}
with $J=\begin{pmatrix}0&1\\-1&0\end{pmatrix}$.
The transfer matrix $T_k$ is given by $T_k=T_{k,r}^{-1}T_k^{\rm w}T_{k,\ell}$.
From $\sigma_k=\sigma_{k-1}+m_k$, we obtain
\begin{equation}\label{TK^A_case_4}
T_k = \left\{
\begin{aligned}
 &{\mathcal C}^{\sigma_{k-1}} T_k^{\rm e} &&\qquad \text{if $m_k$ is even},\\[5pt]
 &i Q {\mathcal C}^{\sigma_{k-1}}\left(\begin{pmatrix}
    -i & 0\\0 & i
\end{pmatrix} T_k^{\rm e}\right) &&\qquad \text{if $m_k$ is odd} (\iff k\in\overline{\LA^{\rm odd}}).
\end{aligned}
\right.
\end{equation}
Here, the matrix $T_k^{\rm e}$ is defined by
\begin{equation*}
    T_k^{\rm e}=\begin{pmatrix}
    \alpha_k&-\overline{\beta_k}\\
    \beta_k&\overline{\alpha_k}
\end{pmatrix},
\end{equation*}
with
\begin{equation*}
    \alpha_k=\alpha_k^{\rm w}+\ord(\mu_1+h),\qquad
    \beta_k=\beta_k^{\rm w}+\ord(\mu_1+h),
\end{equation*}
and ${\mathcal C}$ 
is the complex conjugation operator,
that is ${\mathcal C}\begin{pmatrix}a&b\\c&d\end{pmatrix} = \begin{pmatrix}\bar{a}&\bar{b}\\\bar{c}&\bar{d}\end{pmatrix}$. 
Note that, as $\mu_{m_k}\to\infty$, $T_k$ is asymptotic to an anti-diagonal matrix if $m_k$ is odd, that is, $k\in\overline{\LA^{\rm odd}}$, and to a diagonal matrix otherwise.

\subsubsection{Product of the transfer matrices}\label{sec:prod-Tk}
Let us introduce the notation of the transfer matrix $T_k^\prime$ belonging to $\su (2)$ as follows:
\begin{equation}\label{cases_Tk_2}
T_k^\prime = \left\{
\begin{aligned}
&-iQT_k\qquad&&k\in\overline{\LA^{\rm odd}},\\
& T_k&&\text{otherwise.}
\end{aligned}
\right.
\end{equation}
Then each $T_k'\in\su(2)$ is asymptotic to a diagonal matrix. 
For each $k\in\overline{\LA^{\rm odd}}$, one has $(w_{1,{\ell_k}},w_{2,\ell_k})=(iw_{2,r_k},iw_{1,r_k})T_k^\prime$ whereas one has $(w_{1,{\ell_k}},w_{2,\ell_k})=(w_{1,r_k},w_{2,r_k})T_k^\prime$ for each $k\notin\overline{\LA^{\rm odd}}$.
Recalling the properties of $Q$ \eqref{commute_prop} concerning the commutation and the square, 
we obtain another expression of the scattering matrix using this matrix $T_k^\prime$:
\begin{align}\label{S_Q}
S&= T_r^{-1} T_1T_{1,2}T_2\cdots T_{n-1,n}T_nT_{\ell}
= T_r^{-1} \widetilde{T_1^\prime}\widetilde{T_{1,2}}\widetilde{T_2^\prime}\cdots \widetilde{T_{n-1,n}}\widetilde{T_n^\prime}
 \left(iQ\right)^{n_o}T_\ell, 
\end{align}
where we have defined $\widetilde{T_j'}\in\su (2)$ and $\widetilde{T_{j,j+1}}\in\su (2)$ such that
\begin{equation}\label{wide_tilde}
\widetilde{T_j'} = Q^{l_j} T_j' Q^{l_j} , \qquad 
\widetilde{T_{j,j+1}} =Q^{l_j} T_{j,j+1} Q^{l_j},
\end{equation}
where $l_j$ is the largest integer such that $k(l_j)\le j$ with the convention that $k(0)=0$ ($k(l)$ for $1\le l\le n_o$ is defined by \eqref{def-k(l)}).
Hence the expression \eqref{S_Q} implies that the algebraic lemma \eqref{t21^2_ver1} can be applied directly to the computation of the off-diagonal entry of the product, and that the transition probability depends also on the parity of the number $n_o=\# \overline{\LA^{\rm odd}}$. 

As a sequel to this computation of the product in \eqref{S_Q}, we can derive the dependence on $\muN, \muA$ more precisely.  
Denoting the $(1,1)$-entry of $\widetilde{T_k^{\prime}}$ by $\widetilde{\alpha_k^{\prime}}$, we see that 
$\widetilde{\alpha_k^{\prime}}$ is of $\ord(1)$ 
and, in particular 
$\widetilde{\alpha_k^{\prime}} = 1+\ord(\e_1^2)$ for $k\in\LN$ (see \S\ref{coexist} for the definition of $\e_1$).  
Setting, similarly, the $(2,1)$-entry of $\widetilde{T_k^{\prime}}$ by $\widetilde{\beta_k^{\prime}}$, 
we can rewrite $\widetilde{\beta_k^{\prime}}$ as 
\be
\widetilde{\beta_k^{\prime}} \equiv \left\{
\begin{aligned}
&p_k \muN &&\qquad (k\in \LN),\\
&q_k(m_k, \sigma_{k-1}) \exp [-a_k \muA^{(\mA +1)/\mA}] &&\qquad (k\in \LA), 
\end{aligned}
\right.
\ee
modulo $\ord(\e_1^2)$, 
where $p_k$ and $q_k(m_k, \sigma_{k-1})$ are uniquely determined by \eqref{adi_beta}, \eqref{cases_Tk_2} and \eqref{wide_tilde}.
Notice that $p_k$ and $q_k(m_k, \sigma_{k-1})$ are of $\ord(1)$ in each regime. 
On the other hand, 
$\widetilde{T_{k,k+1}}$ can be regarded as the matrix $T_{k,k+1}$ by replacing ${\tilde V}(t)$ (see \eqref{modi_V}) with $V(t)$. From this fact, it is deduced that the asymptotic of the transition probability in the intermediate regime is determined by the effective energy $\tilde V$. Hence, we can obtain the asymptotic behavior of $|\tau_{21}|^2$ as follows:
\begin{align*}
&\muN^2\left(
\sum_{j\in \LN} \gammaN |v_{j+1}|^{-\frac{2}{\mN +1}} 
  + 2\!\!\!\! \sum_{\substack{j,k\in \LN\\j<k}} \!\!\!\! {\rm Re}\, C_{j,k}^{\NN}(\e,h) 
  \cos \left[
  \frac{1}{h}\int_{t_k}^{t_j} \tilde V (t)dt
  \right]\right)\\
  &\quad + \sum_{k\in \LA} \exp \left[
  -2a_k \muA ^{(\mA +1)/\mA} 
  \right]\\
  &\quad + 2\!\!\!\! \sum_{\substack{j\in \LN,k\in \LA\\j<k}} \!\!\!\! {\rm Re}\, C_{j,k}^{\NA}(\e,h) 
  \muN \exp \left[
  -a_k \muA ^{(\mA +1)/\mA} 
  \right]
  \cos \left[
  \frac{1}{h}\int_{t_k}^{t_j} \tilde V (t)dt
  \right]\\
  &\quad + 2\!\!\! \sum_{\substack{j,k\in \LA\\j<k}} \!\!\! {\rm Re}\, C_{j,k}^{\AAa}(\e,h) 
  \exp \left[
  -(a_j+a_k) \muA ^{(\mA +1)/\mA} 
  \right]
  \cos \left[
  \frac{1}{h}\int_{t_k}^{t_j} \tilde V (t)dt
  \right]\\
  &\quad + \ord(\epsilon_1 \epsilon_2),
\end{align*}
where $\epsilon_1, \epsilon_2$ are given in \S\ref{coexist} and 
\begin{align}\label{form_flafal}
C_{j,k}^{\NN} (\e,h) &= p_j\, \overline{p_k},\\ \label{form_flasha}
C_{j,k}^{\NA} (\e,h) &= 
\left( \prod_{\iota =j+1}^{k-1}\widetilde{\alpha_\iota^{\prime}}^2 \right) \widetilde{\alpha_k^{\prime}}\,
p_j\, \overline{q_k(m_k, \sigma_{k-1})},\\ \label{form_shasha}
C_{j,k}^{\AAa} (\e,h) &= 
\widetilde{\alpha_k^{\prime}}
\left( \prod_{\iota =j+1}^{k-1}\widetilde{\alpha_\iota^{\prime}}^2 \right) \widetilde{\alpha_k^{\prime}}\,
q_j(m_j, \sigma_{j-1}) \,\overline{q_k(m_k, \sigma_{k-1})}.
\end{align}
The proof of Theorem~\ref{2ndthm} has been completed.

\section*{acknowledgments}
The authors gratefully express our thanks to M.~Zerzeri for useful discussions. 
This research is supported by the Grant-in-Aid for JSPS Fellows Grant Number JP22KJ2364, Grant-in-Aid for Research Activity Start-up Grant Number JP25K23332, and Grant-in-Aid for Scientific Research (C) JP18K03349, JP21K03282. The authors are grateful to the support by the Research Institute for Mathematical Sciences, an International Joint Usage/Research Center located in Kyoto University.

\section*{Data Availability Statement}
There is no data appearing in this paper.

\appendix
\section{Jost solutions}\label{app_jost} 


In this subsection \ref{app_jost}, we give the existence of the Jost solutions for the definition of the scattering matrix. We remark that the smallness of $h$ is not required for the argument here.

We first consider the Jost solutions $J_r^\pm$ near $+\infty$. A discussion for $J_l^\pm$ is done similarly but the difference is that we are assuming that $V_r$ is positive (Condition~\ref{condi_1}).
Let $H_r$ denote the limiting Hamiltonian at $+\infty$:
\be
H_r:=\begin{pmatrix}
V_r&\e\\\e&-V_r
\end{pmatrix}. 
\ee
The functions defined by
\be\label{def_pphi}
\pphi_r^+(t) = e^{-i \lambda_r t/h} 
\begin{pmatrix}
    \cos \theta_r \\ \sin \theta_r
\end{pmatrix},\quad 
\pphi_r^-(t) = e^{+i \lambda_r t/h} 
\begin{pmatrix}
    -\sin \theta_r \\ \cos \theta_r
\end{pmatrix},
\ee
where $\lambda_r = \sqrt{V_r^2 +\e^2}$ and $\tan 2\theta_r = \e/V_r$ $(0 < \theta_r < \pi/4)$, are particular solutions to $hD_t \psi + H_r \psi =0$ and form a basis of $\C^2$ for each $t\in \R$. 

\begin{proposition}\label{exist:Jost}
There uniquely exists a pair of solutions $(\phi_r^+,\phi_r^-)$ to the system \eqref{eq:OurEq} such that
\be\label{eq:JostSol}
\lim_{t\to+\infty} \left(\phi_r^\pm(t)-\pphi_r^\pm(t) \right) =0.
\ee
\end{proposition}

\begin{proof}
Let $U(t)$ be a $2\times 2$-matrix valued $C^1$-function. We have
\ben
\frac d{dt}(\Phi_r(t)U(t))=\Phi_r '(t)U(t)+\Phi_r(t) U'(t)
=\frac 1{ih} H_r \Phi_r (t)U(t)+\Phi_r(t) U'(t)
\een
with $\Phi_r :=(\pphi_r^+,\pphi_r^-)$. Thus, if $U(t)$ satisfies 
\be\label{eq:diffU}
U'=\frac 1{ih} \Phi_r^{-1}(H-H_r)\Phi_r U,
\ee
each column of the matrix-valued function $\Phi_r U$ is a solution to the equation \eqref{eq:OurEq}.  
Put $B_r(t):=\Phi_r^{-1}(H-H_r)\Phi_r$. 
From the identity  
\be
H-H_r =(V(t)-V_r)
\begin{pmatrix}
1&0\\0&-1
\end{pmatrix}
\ee
and Condition~\ref{condi_1}, the matrix-valued function $B_r(t)$:
\be\label{A_express}
B_r(t) = (V(t) - V_r) 
\begin{pmatrix}
\cos 2\theta_r & -e^{+2i \lambda_r t/h} \sin 2\theta_r\\
- e^{-2i \lambda_r t/h} \sin 2\theta_r & - \cos 2\theta_r
\end{pmatrix},
\ee
is integrable on the half-line $[0,\infty[$. 
Then the function
\be\label{def-Ur}
U_r(t):=\exp\left(-\frac ih\int_{+\infty}^t B_r(s)ds\right),
\ee
is well-defined and solves the equation \eqref{eq:diffU} with the boundary condition 
\ben
\lim_{t\to+\infty}U_r(t) = \mathrm{Id}.
\een
Here we recall that ${\rm Id}$ stands for the $2 \times 2$ unit matrix. 
We finally obtain the solutions $\phi_r^\pm(t)$ with the asymptotic behavior \eqref{eq:JostSol}:
\be
(\phi_r^+(t),\phi_r^-(t)):=\Phi_r(t) U_r(t).
\ee
\end{proof}

From Proposition \ref{exist:Jost} and the trace-free property of $H(t;\e)$, the pair of $(\phi_r^+, \phi_r^-)$ forms a basis. Similarly, this fact 
implies that $\phi_r^+$ (resp. $\phi_r^-$) coincides with the Jost solution $J_r^+$ (resp. $J_r^-$). 
\bigskip

Next, we give the asymptotic behaviors of $\phi^\pm_r$ as $\e \to 0$ near some fixed point $t_r$. 
Take $t_r>t_1$ (recall that $t_1=\max\{t\in \R\, ;\,V=0\}$ is the first zero of $V$)
satisfying
\be\label{non-zero-condition2}
\int_{+\infty}^{t} (V(s)-V_r)\,ds\neq0.
\ee

put 
\ben
R_r=V_r t_r+\int_{+\infty}^{t_r}(V(s)-V_r)\,ds,\qquad
u_{r}^\pm =\exp\left(\mp \frac ih \int_{t_r}^t V(s)ds\right).
\een
\begin{proposition}\label{Psi_asym_e}
There exists a neighborhood of $t_r$ such that we have
\begin{align*}
\phi_r^+(t)=e^{-i R_r/h}
\begin{pmatrix}
u_{r}^+ +\ord(\e^2/h)\\\ord(\e)
\end{pmatrix},\quad
\phi_r^-(t)=e^{+i R_r/h}
\begin{pmatrix}
\ord(\e)\\u_{r}^-
+\ord(\e^2/h)
\end{pmatrix}
\end{align*}
as $(\e^2/h,\e)\to(0,0)$ uniformly with respect to $t$.
\end{proposition}

Before proving Proposition~\ref{Psi_asym_e}, we prepare the following.
\begin{lemma}\label{alg_UA} 
Let $B$ be a matrix of the form:
\ben
B= \frac{i}{h}
\begin{pmatrix}
-a & b\\ \bar{b} & a
\end{pmatrix}
\een
with $a\in \R\setminus\{0\}$ and $b\in \C$. 
For $|b/a|\ll1$, one has
\begin{align*}
e^{B} &= \begin{pmatrix}
    e^{-ia/h} + \ord((b/a)^2) & \ord(b/a)\\
    \ord(b/a) & e^{+ia/h} + \ord((b/a)^2)
\end{pmatrix}. 
\end{align*}
\end{lemma}

\begin{proof}
Since $B^2 = -h^{-2}(a^2 + |b|^2) \mathrm{Id}$, an algebraic computation gives 
\be
e^B=
\left(\cos \frac{z}{h}\right)\ope{Id}+\frac iz \left(\sin \frac zh\right)\begin{pmatrix}
    -a&b\\\bar{b}&a
\end{pmatrix}
\ee
where $z = \sqrt{a^2 + |b|^2}$. 
We have $z=\ope{sgn}(a)a(1+\ord(b^2/a^2))$ under $|b/a|\ll1$ and $a\neq0$. This gives the following asymptotic formula:
\ben
e^{B} = \begin{pmatrix}
    e^{-ia/h} + \ord((b/a)^2) & 0\\
    0 & e^{+ia/h} + \ord((b/a)^2)
\end{pmatrix}
+\frac iz \left(\sin\frac zh\right)\begin{pmatrix}
    0&b\\\bar{b}&0
\end{pmatrix}.
\een
The lemma follows from $b/z=\ord(b/a)$.
\end{proof}

\begin{proof}[Proof of Proposition~\ref{Psi_asym_e}]
From the expression \eqref{A_express}, $U_r$ defined by \eqref{def-Ur} is written as
\be
U_r(t;\e,h) =\exp\left(\frac ih\begin{pmatrix}
-\mc{I}_r(t)\cos2\theta_r(\e)&\mc{J}_r(t;h)\sin2\theta_r(\e)\\[7pt]
\overline{\mc{J}_r(t;h)}\sin2\theta_r(\e)&\mc{I}_r(t)\cos2\theta_r(\e)
\end{pmatrix}\right),
\ee
where 
\ben
\begin{aligned}
&
\mc{I}_r(t)=\int_{+\infty}^t(V(s)-V_r)ds,\quad &&\mc{J}_r(t;h)=\int_{+\infty}^t(V(s)-V_r)e^{+2is\lambda_r/h}ds.
\end{aligned}
\een
Apply Lemma~\ref{alg_UA} with
\be
a =a(t,\e) = \mc{I}_r(t)\cos2\theta_r(\e),\quad b=b(t,\e,h) = \mc{J}_r(t;h)\sin2\theta_r(\e).
\ee
By the choice of $t_r$ with the condition \eqref{non-zero-condition2}, $a=\mc{I}_r(t)\cos2\theta_r(\e)$ never vanishes for $t$ near $t_r$.
By definition, we have $\theta_r(\e)=\ord(\e)$, and consequently $|b/a|=\ord(\e)$. Then Lemma~\ref{alg_UA} shows
\be\label{U_asym_e}
U_r(t;\e,h)=\begin{pmatrix}
    e^{-i\mc{I}_r(t)/h} +\ord(\e^2/h)&  \ord(\e)\\
    \ord(\e)&  e^{+i\mc{I}_r(t)/h} +\ord(\e^2/h)
\end{pmatrix}
\ee
Note that we have the following decomposition of $\mc{I}_r(t)$:
\be\label{decomp-Ir}
\mc{I}_r(t) = \mc{I}_r(t_r) + \int_{t_r}^t (V(s)-V_r)ds = R_r + \int_{t_r}^t V(s) ds - V_rt.
\ee
Since $\Phi_r(t)$ admits the asymptotic formula
\be\label{Phi_asym_e}
\Phi_r(t) = \begin{pmatrix}
    \pphi_r^+ & \pphi_r^-
\end{pmatrix} =
\left( 1+ \ord \left(\frac{\e^2}{h}\right)\right)
\begin{pmatrix}
    e^{-i V_r t/h} & \ord(\e)\\
    \ord(\e) & e^{+i V_r t/h}
\end{pmatrix},
\ee
as $(\e^2/h,\e)\to(0,0)$, Proposition~\ref{Psi_asym_e} follows from \eqref{U_asym_e} and \eqref{decomp-Ir}.
\end{proof}

The asymptotic formula
\be\label{fomula_Tr_01}
\begin{pmatrix} J_r^+ & J_r^- \end{pmatrix}
= \begin{pmatrix}
     u_r^+ + \ord(\e^2/h) & \ord(\e)\\
    \ord(\e) &  u_r^- + \ord( \e^2 /h)
\end{pmatrix}
\begin{pmatrix}
    e^{-i R_r/h} & 0\\
    0 & e^{+i R_r/h}
\end{pmatrix}
\ee
is directly deduced from Propositions~\ref{exist:Jost} and \ref{Psi_asym_e}. 
Therefore, we obtain
\be\label{fomula_Tr_02}
T_r = \begin{pmatrix}
    e^{-i R_r/h}+\ord(\e^2/h) & \ord(\e^2)\\
    \ord(\e^2) & e^{+i R_r/h}+\ord(\e^2/h)
\end{pmatrix}.
\ee

For the Jost solutions $J_\ell^\pm (t)$, the same argument as above works when $V_\ell >0$ and a similar one induces the existences and the asymptotic behaviors when $V_\ell <0$. 

In the case where $V_\ell >0$, by exchanging the sub-index $r$ for $\ell$, 
one sees
\be\label{fomula_Tell_02}
T_\ell = \begin{pmatrix}
    e^{-i R_\ell/h}+\ord(\e^2/h) & \ord(\e^2)\\
    \ord(\e^2) & e^{+i R_\ell /h}+\ord(\e^2/h)
\end{pmatrix}. 
\ee
Here 
\ben
R_\ell = V_\ell t_\ell+\int_{-\infty}^{t_\ell}(V(s)-V_\ell)\,ds
\een
with $t_\ell < t_n=\min\{t\in \R\, ;\,V=0\}$ satisfying that the second integral term in the right-hand side does not vanish. 

In the case where $V_\ell <0$, we choose instead of \eqref{def_pphi} particular solutions $\pphi_\ell^\pm(t)$ to $hD_t \psi + H_\ell \psi = 0$ as 
\be\label{def_pphi_nega}
\pphi_\ell^+(t) = e^{-i \lambda_\ell t/h} 
\begin{pmatrix}
    \sin \eta_\ell \\ \cos \eta_\ell
\end{pmatrix},\quad 
\pphi_\ell^-(t) = e^{+i \lambda_\ell t/h} 
\begin{pmatrix}
    -\cos \eta_\ell \\ \sin \eta_\ell
\end{pmatrix},
\ee
where $\lambda_\ell = \sqrt{V_\ell^2 +\e^2}$ and $\tan 2\eta_\ell = \e/(-V_\ell)$ $(0 < \eta_\ell < \pi/4)$. 
They coincide with the leading terms of Jost solutions $J_\ell^\pm(t)$ when $V_\ell < 0$ and satisfy the asymptotic formulas:
\ben
\pphi_\ell^+ \sim e^{+i V_\ell t/h} \begin{pmatrix}\ord(\e) \\ 1+ \ord(\e^2/h)\end{pmatrix},\quad 
\pphi_\ell^- \sim e^{-i V_\ell t/h} \begin{pmatrix}-1+ \ord(\e^2/h)\\ \ord(\e)\end{pmatrix}
\een
as $(\e^2/h,\e) \to (0,0)$ for each $t$. 
One sees that, with \eqref{def_pphi_nega}, Proposition~\ref{exist:Jost} also holds. 
One also have similar asymptotic formulas to those of Proposition~\ref{Psi_asym_e}:
\begin{align*}
\phi_\ell^+(t)=e^{+i R_\ell/h}
\begin{pmatrix}
\ord(\e)\\
u_{\ell}^+
+\ord(\e^2/h)
\end{pmatrix},\quad
\phi_\ell^-(t)=e^{-i R_\ell/h}
\begin{pmatrix}
-u_{\ell}^-
+\ord(\e^2/h)\\
\ord(\e)
\end{pmatrix}
\end{align*}
as $(\e^2/h,\e) \to (0,0)$ uniformly in a small neighborhood of $t= t_\ell$, where $u_\ell^\pm = \exp(\mp i\int_{t_\ell}^tV(s)ds/h)$. 
We obtain
\be\label{fomula_Tell_03}
T_\ell = \begin{pmatrix}
    \ord(\e) & -e^{-i R_\ell/h}+\ord(\e^2/h) \\
    e^{+i R_\ell/h} +\ord(\e^2/h)& \ord(\e)
\end{pmatrix}. 
\ee

\section{Review of the exact WKB method}\label{Review_WKB}

This section is devoted to a quick review of the exact WKB method under the notations used in this paper. 
The exact WKB method (sometimes called the complex WKB method) was initiated by G\'erard-Grigis \cite{GeGr88_01} 
and developed to a first order $2\times 2$ system by Fujii\'e-Lasser-N\'ed\'elec \cite{FuLaNe09_01}.

Put $\psi(t;h):= \frac{1}{2}\begin{pmatrix}1&i\\i&1\end{pmatrix}\phi(t;h)$,
where $\phi(t;h)$ is the solution of \eqref{eq:OurEq}.
Then the original equation \eqref{eq:OurEq} can be reduced to 
 the following first order $2\times 2$ system:
\begin{equation}\label{wkbModel}
\frac{h}{i} \frac{d}{dt} \phi(t;h) = 
\begin{pmatrix}
0 & \alpha(t)\\ -\beta(t) & 0
\end{pmatrix}
\phi(t;h),
\end{equation}
where $\alpha(t) = -iV(t) - \e$ and $\beta(t) = -iV(t) + \e$. 
In this section, we suppose that $\e$ is small and fixed. 
We treat this equation \eqref{wkbModel} on some simply connected domain ${\mathcal S}\subset \C$, where $V(t)$ is analytic and vanishes only at $t=t_0$.   
We define for any fixed point $a\in {\mathcal S}$
$$
z_a(t) = \int_a^t \sqrt{\alpha(s)\beta(s)}\,ds = i \int_a^t \sqrt{V(s)^2 + \e^2}\, ds,
$$
where the branch of the integrand $\sqrt{V(t)^2 + \e^2}$ is taken $\e$ at $t=t_0$. 
Notice that for any $a, \tilde a \in {\mathcal S}$ one has 
\begin{equation}\label{z_affine}
z_a(t) = z_{\tilde a}(t) + \int_a^{\tilde a} \sqrt{\alpha(s)\beta(s)} ds.
\end{equation}
The set of turning points which are zeros of $\alpha(t)\beta(t)$ is denoted by $\Lambda$, and 
a simply connected sub-domain of ${\mathcal S}\setminus \Lambda$ is done by $\widetilde {\mathcal S}$. 
Remark that the mapping $z_a$ is bijective from $\widetilde {\mathcal S}$ to $z_a(\widetilde {\mathcal S})$. 
Making a branch cut from each turning point to the infinity suitably, we may suppose that $\widetilde {\mathcal S}$ includes the real interval near $t_0$. 
This fact permits us to know that $ \text{Re}\,z_a(t)$ increases as $ \text{Im}\,t$ decreases near a complex neighborhood of $t_0$. These properties are crucial to the following Lemma \ref{wkblem2} and Lemma \ref{wkblem3}.  

In the context of the exact WKB method, we  consider the solution of \eqref{wkbModel}, 
$\phi(t;h)$, as a function of the variable $z$ by setting
$\phi^\pm(t;h) = e^{\pm z/h} M^\pm(z) w^\pm(z;h)$, with  
$$
M^\pm(z) = \begin{pmatrix}
K(z)^{-1} & K(z)^{-1}\\
\mp i K(z) & \pm i K(z)
\end{pmatrix}, \quad 
K(z(t)) = \left(\frac{\beta(t)}{\alpha(t)}\right)^{\frac{1}{4}} 
= \left(\frac{-iV(t) + \e}{-iV(t) -\e}\right)^{\frac{1}{4}}.
$$
Notice that $K(z(t))$ is independent of the base point $a$ involved in the definition of the function $z(t)=z_a(t)$.
The branch of $K(z(t))$ is taken $e^{-i\pi/4}$ at $t=t_0$ with 
the branch cut along a positive real axis on $\C_z$.
Here the vector-valued functions $w^\pm(z;h) $ are determined as solutions of 
$$
\frac{d}{dz} w^\pm (z;h) = \begin{pmatrix}
0 & \frac{\partial_z K(z)}{K(z)}\\ \frac{\partial_z K(z)}{K(z)} & \mp \frac{2}{h}
\end{pmatrix}
w^\pm(z;h).
$$
Moreover, by the identity \eqref{z_affine} and the following equality
\begin{equation}\label{sing_K'/K}
\frac{\partial_z K(z(t))}{K(z(t))} = 
\frac{\alpha(t)\beta'(t) - \alpha'(t)\beta(t)}{4(\alpha(t)\beta(t))^{3/2}}\,,
\end{equation}
we see that $\partial_z K(z(t))/K(z(t))$ and $w^\pm(z(t);h)$ are independent of $a$. 
The above equality \eqref{sing_K'/K} implies that 
the function $\partial_z K(z)/K(z)$ has a simple pole at $z=z(\zeta)$ for each turning point $\zeta$.
Note that, in our case, each turning point is simple (i.e. a simple zero of $\alpha(t)\beta(t)$) for $\e>0$ small enough.

\medskip

\noindent
Generally, even if the vector-valued symbols $w^\pm(z;h)$ are developed with respect to $h$ small enough, 
the series do not converge. 
The essential idea of the work \cite{GeGr88_01} (see also the work \cite{FuLaNe09_01}) is to introduce a resummation  
by using the following integral recurrence system on $\C_z$. 
More precisely, for any $b\in \widetilde {\mathcal S}$, the vector-valued functions $w^\pm(z;h)=w^{\pm}(z,z(b);h)$ are of the form: 
\begin{equation}\label{resum1}
w^{\pm}(z,z(b);h) = \sum_{k\geq 0} w^{\pm}_{k}(z,z(b);h) 
 = \sum_{k\geq 0}\left(\begin{matrix}
w^{\pm}_{2k}(z,z(b);h)\\w^{\pm}_{2k-1}(z,z(b);h)
\end{matrix}\right),
\end{equation}
where the sequences $\big\{w^{\pm}_{k}(z,z(b);h)\big\}_{k\in\N}$ are defined by 
\begin{equation*}
\left\{
\begin{aligned}
w^{\pm}_0(z,z(b);h) &\equiv \, 1, \quad w^{\pm}_{-1}(z,z(b);h) \equiv \, 0,\\
\ w^{\pm}_{2k+1}(z,z(b);h) &=\int_{z(b)}^z 
e^{\pm\frac{2}{h}(\zeta-z)}\frac{ \partial_z  K(\zeta)}{K(\zeta)}
w^{\pm}_{2k}(\zeta,z(b);h)\,d\zeta  &&(k\geq 0),\\
w^{\pm}_{2k}(z,z(b);h) &=\int_{z(b)}^z 
\frac{\partial_z K(\zeta)}{K(\zeta)} w^{\pm}_{2k-1}(\zeta,z(b);h)\,d\zeta &&(k\geq 1).
\end{aligned}
\right.
\end{equation*}

Thanks to the above resummation, the vector-valued symbol expansions \eqref{resum1} converge
absolutely and uniformly in a neighborhood of $z(b)$ for $b\in \widetilde {\mathcal S}$ (see, for example, the work \cite[Lemma 3.2]{FuLaNe09_01}). 
Hence, for any fixed $(a,b)\in {\mathcal S}\times \widetilde {\mathcal S}$, we can define 
the exact WKB solutions of type $\pm$ as follows:
\begin{align}\label{def_exactWKBsol}
\psi^{\pm}(t,a,b;h)
= \frac{1}{2}\begin{pmatrix}1&i\\i&1\end{pmatrix} e^{\pm z_a(t)/h} M^{\pm}(z(t))w^{\pm}(z(t),z(b);h), 
\end{align}
which are linearly independent exact solutions of \eqref{eq:OurEq}.
Notice that $a\in {\mathcal S}$ is the base point of the phase and $b\in \widetilde {\mathcal S}$ is that of the symbol.\\

We conclude this subsection by recalling some results concerning the exact WKB solutions given by \eqref{def_exactWKBsol}.
In fact, the exact WKB method is based on two properties, which are the Wronskian formula 
between the exact WKB solutions of type $\pm$ and the asymptotic expansion with respect to $h$ of the symbol.

\begin{lemma}{\rm \cite[\bf Proposition 2.2.2]{Wa06_01}}\label{wkblem1}
The Wronskian between any exact WKB solutions of type $\pm$ 
with the same base point of the phase satisfies:
\begin{equation}
{\mathcal W} [\psi^+(t,a,b_+;h), \psi^-(t,a,b_-;h)] = 2i \sum_{k\geq 0} w^{+}{2k}(z(b_-),z(b_+);h),
\end{equation}
where $a\in {\mathcal S}$ and $b_{\pm}\in \widetilde {\mathcal S}$.
Here the Wronskian between $\C^2$-valued functions $\psi_1$ and $\psi_2$  
is defined by ${\mathcal W}[\psi_1,\psi_2] := {\rm det}\left( \psi_1 \psi_2\right)$.  
\end{lemma}

The proof of this lemma is based on a direct computation and 
the independence of the Wronskian with respect to the variable $t$ 
thanks to the trace-free matrix in \eqref{wkbModel}. 
The prefactor $2i$ is exactly $\det M_+$. \\

To state the next result, we introduce {\it canonical curves of type $\pm$} in $\widetilde {\mathcal S}$ from a fixed point $b$ to $t$ 
along which $\pm \text{Re}\,z_a(t)$ increase strictly, for a fixed $a\in {\mathcal S}$. 
The advantage of the integral recurrence system is to give not only an absolutely convergence 
but also $\C^2$-valued asymptotic sequences with respect to $h$ uniformly away from turning points. 
More precisely,

\begin{lemma}{\rm \cite[\bf Proposition 2.3.1]{Wa06_01}}\label{wkblem2}
If there exist canonical curves of type $\pm$ from $b_\pm$ to $t$ denoted by $\gamma_\pm$,  
then the vector-valued symbols have the following asymptotic expansions:
\begin{equation}\label{symbolasymptotic}
w^{\pm}(z(t),z(b_\pm);h) = \begin{pmatrix}
1\\
0
\end{pmatrix}
\left( 1+ {\mathcal O}\left( \frac{h}{{\rm dist}(\gamma_\pm;\Lambda) }\right) \right)
\end{equation}
as $h$ tends to $0$, where ${\rm dist}(\gamma_\pm;\Lambda)$ stands for $\underset{t\in \gamma_\pm, \zeta\in \Lambda}{\inf}|z_{\zeta}(t)|$. 
\end{lemma}

This lemma can be proved by an integration by parts thanks to the exponential decaying along the canonical curve. \\

Combining Lemmas \ref{wkblem1} and \ref{wkblem2}, we obtain the asymptotic expansion of the Wronskian:
\begin{lemma}{\rm \cite[\bf Proposition 2.4.1]{Wa06_01}}\label{wkblem3}
If there exists a canonical curve of type $+$ from $b_+$ to $b_-$ denoted by $\gamma$, 
the Wronskian between any exact WKB solutions of type $\pm$ with the same base point of the phase  
has the following asymptotic expansion, 
\begin{equation}\label{wkbWron}
{\mathcal W} [\psi^+(t,a,b_+;h), \psi^-(t,a,b_-;h)] =  
2i + {\mathcal O}\left( \frac{h}{{\rm dist}(\gamma;\Lambda) }\right)
\end{equation}
as $h$ tends to $0$, where ${\rm dist}(\gamma;\Lambda)=\underset{t\in \gamma, \zeta\in \Lambda}{\inf}|z_{\zeta}(t)|$. 
\end{lemma}

\section{Algebraic lemma}\label{app_alg_lem}

In order to know the asymptotic behavior of the scattering matrix \eqref{formula_Smatrix}, 
it suffices to compute the products of the matrices of the following forms modulo $\ord(\mu^2)$ with a small parameter $\mu$:
\ben
T_k(\mu) \equiv \begin{pmatrix}
\alpha_k(\mu) & -\oline{\beta_k(\mu)}\\ \beta_k(\mu) & \overline{\alpha_k(\mu)}
\end{pmatrix}, \qquad 
T_{k,k+1} \equiv \begin{pmatrix}
\upsilon_k & 0\\ 0 & \overline{\upsilon_k}
\end{pmatrix}, 
\een
where $\alpha_k$, $\beta_k$, and $\upsilon_k$ are complex numbers such that $\det T_k=\det T_{k,k+1}=1$, namely $|\alpha_k|^2+|\beta_k|^2=|\upsilon_k|^2=1$, and $\beta_k(\mu)=\ord(\mu)$ as $\mu\to0$.
Notice that, in our context of Theorem~\ref{mainthm} (see \eqref{eq:Asym-T}), $T_k$ and $T_{k,k+1}$ have this form with the numbers given modulo $\ord(\mu^2)$ by 
\begin{align*}
\alpha_k \equiv 1,
\quad \beta_k \equiv -i \overline{\omega_k} \mu_{m_k} = \ord(\mu_{m_k}),\quad 
\upsilon_k \equiv \exp\left(-\frac{i}{h} \int_{t_{k+1}}^{t_k} V(t)dt \right). 
\end{align*}
In this subsection we give an algebraic formula by means of these notations $\alpha_k$, $\beta_k(\mu)$ and $\upsilon_k$ for simplicity. 
We know that the product of them is of the form
\be\label{prod_TT}
T_k T_{k,k+1} = \begin{pmatrix}
\alpha_k \upsilon_k& -\overline{\beta_k  \upsilon_k}\\ \beta_k \upsilon_k & \overline{\alpha_k \upsilon_k}
\end{pmatrix}.
\ee
Let $\su (2)$ be the special unitary group of degree $2$ given by
\be\label{set_M}
\su (2) = \left\{ T \in M_2(\C)\, ;\, T = \begin{pmatrix} a & - \bar{b}\\ b & \bar{a}\end{pmatrix}, a,b\in \C,\ \det T=1 \right\}.
\ee
Note that the above definition of $\su (2)$ is equivalent to the standard one: $\su (2)=\{T\in\mathrm{U}(2)\,;\,\det T=1\}$, where $\mathrm{U}(2)$ is the unitary group of degree $2$.
One sees that all of the above matrices belong to $\su (2)$.  
Denoting the products of these matrices by 
\be\label{notation_TT}
{\mathcal T}=  
T_1T_{1,2}T_2\cdots T_n T_{n,n+1}
= 
\begin{pmatrix}
\tau_{11} & \tau_{12}\\ \tau_{21} & \tau_{22}
\end{pmatrix} \in \su (2), 
\ee
we get the following lemma:
\begin{lemma}\label{lem_alg_TT}
As $\mu \to 0$, the following asymptotic formulas hold. 
\begin{align}
    &\tau_{11}(\mu) = \prod_{j=1}^{n} \alpha_j \upsilon_j + \ord (\mu^2),\\ 
    &\tau_{21}(\mu) = \sum_{j=1}^n \left(\prod_{\iota=1}^{j-1} \overline{\alpha_\iota} \right)
    \beta_j(\mu) 
    \left(\prod_{k=j+1}^{n} \alpha_k \right)
    \left(\prod_{\iota=1}^{j-1} \overline{\upsilon_\iota} \right)\left(\prod_{k=j}^{n} \upsilon_k \right)
    + \ord (\mu^2),\\[7pt]
    &\tau_{12}(\mu) = -\overline{\tau_{21}}(\mu), \qquad \tau_{22}(\mu) = \overline{\tau_{11}}(\mu),
\end{align}
with the convention that $\prod_{\iota=1}^0 \overline{\alpha_\iota}=\prod_{\iota=1}^0 \overline{\upsilon_\iota}=1$.
\end{lemma}
The proof of this lemma is based on the mathematical induction for the product of the matrices \eqref{prod_TT}. 
A simple computation of $|\tau_{21}|^2$, which corresponds to the transition probability, gives
\begin{align}\nonumber
|\tau_{21}|^2 &= \sum_{j=1}^n |\beta_j(\mu)|^2 
+ 2 \re\, \left[\sum_{1\leq j<k \leq n} \beta_j(\mu)\alpha_j \left(\prod_{\iota =j+1}^{k-1} \alpha_\iota^2 \right) \alpha_k \overline{\beta_k}(\mu) 
\left(\prod_{\iota =j}^{k-1} \upsilon_\iota^2 \right) \right]
\\ \label{t21^2_ver1}
 &\qquad 
 +\ord(\mu^3),
\end{align}
with the convention that $\prod_{\iota=j+1}^j \alpha_\iota^2=\prod_{\iota=j}^{j-1} \upsilon_\iota=1$.
Note that we used $|\alpha_k| = 1+\ord(\mu^2)$ and $|\upsilon_k|= 1$ in the above computation. In particular, when $\alpha_k =1+\ord(\mu^2)$, we have
\be\label{t21^2_ver2}
|\tau_{21}|^2 = \sum_{j=1}^n |\beta_j(\mu)|^2 + 2 \re\, \left[\sum_{1\leq j<k \leq n} \beta_j(\mu) \overline{\beta_k(\mu)} \left(\prod_{\iota =j}^{k-1} \upsilon_\iota^2 \right)
\right]+\ord(\mu^3).
\ee


\end{document}